\documentclass[12pt]{article}
\makeatletter
\newcommand{\singlespacing}{\let\CS=\@currsize\renewcommand{\baselinestretch}{1}\tiny\CS}
\usepackage{epsfig}
\usepackage[utf8]{inputenc}
\usepackage{amsmath}
\usepackage{amsthm}
\usepackage{mathptmx}
\usepackage{authblk}
\usepackage{adjustbox}
\usepackage{subfig}
\usepackage{xcolor}
\usepackage{xspace}
\usepackage[colon,sort&compress,round,authoryear]{natbib}
\usepackage{amsfonts}
\usepackage{tgpagella}
\usepackage{caption}
\DeclareCaptionFont{xipt}{\fontsize{11}{9}\mdseries}
\usepackage[font=xipt,labelfont=bf]{caption}

\usepackage[colorlinks=true,urlcolor=black,citecolor=blue,linkcolor=blue,bookmarks=true]{hyperref}
\usepackage{float}

\usepackage[parfill]{parskip}

\newcommand{\be}{\begin{equation}}
\newcommand{\ee}{\end{equation}}
\newcommand{\beanno}{\begin{eqnarray*}}

\newcommand{\eeanno}{\end{eqnarray*}}
\newcommand{\bea}{\begin{eqnarray}}
\newcommand{\eea}{\end{eqnarray}}
\newcommand{\ba}{\begin{array}}
\newcommand{\ea}{\end{array}}

\newcommand{\bc}{\begin{center}}
\newcommand{\ec}{\end{center}}

\newcommand{\ds}{\displaystyle}

\renewcommand{\baselinestretch}{1.40}
\newtheorem{theorem}{Theorem}[section]

\newtheorem{lemma}{Lemma}[section]

\allowdisplaybreaks
\begin{document}
\title{\sc A new decision theoretic sampling plan for type-I  and type-I hybrid  censored samples from the exponential distribution}
\author{Deepak Prajapati$^{1}$, Sharmistha Mitra$^{1}$ and Debasis Kundu$^{1,2}$} 
\date{}
\maketitle
\begin{abstract}
The study proposes a new decision theoretic sampling plan (DSP) for Type-I and  Type-I hybrid censored samples when the lifetimes of individual items are exponentially distributed with a scale parameter. The DSP is based on an estimator of the scale 
parameter which always exists, unlike the MLE which may not always exist. Using a quadratic loss function and a decision function based on the proposed estimator, a DSP is derived. To obtain the optimum DSP, a finite algorithm is used. Numerical results demonstrate  that in terms of the Bayes risk, the optimum DSP is as good as the Bayesian sampling plan (BSP) proposed by \cite{lin2002bayesian} and \cite{liang2013optimal}.  The proposed DSP performs better than the sampling plan of \cite{Lam1994bayesian} and \cite{lin2008-10exact} 
in terms of Bayes risks.  The main advantage of the proposed DSP is that for higher degree polynomial and non-polynomial loss functions, it can 
be easily obtained as compared to the BSP.  
\end{abstract}
\noindent {\sc Keywords and Phrases:} Exponential distribution, Type-I and Type-I hybrid censoring, Decision theoretic sampling plan, Bayes risk, Bayesian sampling plan.\\
\noindent\rule{16.5cm}{0.4pt}\\
\noindent$^1$ Department of Mathematics and Statistics, Indian Institute of Technology Kanpur, India. \\
\noindent$^2$ Corresponding author e-mail: kundu@iitk.ac.in

\section{\scshape Introduction}
The sampling plan is an important instrument of any quality control experiment, which is used to test the quality of batch of items. 
A good sampling plan is important for manufacturers because a batch of items manufactured by them  at the acceptable level of quality 
will have a good chance to be accepted by the plan. In the decision-theoretic approach, a sampling plan is determined by making an optimal decision on the basis of maximizing the return or minimizing the risk. So, for the economical point of view, it is more reasonable and realistic approach and 
therefore, it is widely employed by many statisticians. An extensive amount of work has been done along this line, see, for example,  \cite{hald}, \cite{FM1974}, \cite{yeh1988}, \cite{Lam1994bayesian}, \cite{lin2002bayesian}, \cite{Huang2002-4lin}, \cite{Chenchou2004},     \cite{lin2008-10exact}, \cite{liang2013optimal},\cite{Tsai2014}, and \cite{yang2015optimal}.

 In most of the life testing experiments, censoring is inevitable, i.e., the experiment terminates before all the experimental items fail.
As a common practice, we put $n$ items on test and terminate the test when a preassigned $r$ number of items fail. This is known as the 
Type-II censoring, which ensures $r$ number of failures. But, in this case the experimental time would be unusually long for high quality items. To tackle this problem, the Type-I censoring scheme is used, in which we put $n$ items on test and terminate the test at a preassigned time $\tau$, no matter how many failures happen before the time $\tau$. \cite{Lam1994bayesian} has provided a Bayesian sampling plan 
for a Type-I censoring scheme based on a suitable  decision function and when the loss function is quadratic. \cite{lin2002bayesian} has proved that Lam's sampling plan is neither optimal nor Bayes and they have provided a Bayesian sampling plan in this case. 
\par
The hybrid censoring is more economical and logical because it combines the advantages of both types of censoring.  In the Type-I hybrid censoring the experiment is terminated at the time $\tau^*=min\{X_{(r)},\tau\}$, where $\tau$ is a fixed time and $X_{(r)}$ is the time to the $r^{th}$ failure. In a Type-II hybrid censoring the experiment is terminated at the time $\tau^*=max\{X_{(r)},\tau\}$. 
\cite{lin2008-10exact} derived an optimal sampling plan for both hybrid censoring schemes using the Bayesian approach. \cite{liang2013optimal} found the exact Bayes decision function and derived an optimum  Bayesian sampling plan for the Type-I hybrid censoring based on a 
quadratic loss function. An extensive amount of literature is available on all the above sampling plans which are decision theoretic in nature and are based on the estimator of the mean lifetime of the exponential distribution.
\par
 In this paper, we develop a decision theoretic sampling plan (DSP) for Type-I and  Type-I hybrid censored samples using a decision function which is based on a suitable estimator of $\lambda$ in place of the estimator of the mean lifetime $\ds \theta=\frac{1}{\lambda}$.  We consider the sampling plans $(n, \tau,\zeta)$ under the Type-I censoring and  $(n,r,\tau,\zeta)$ under the Type-I hybrid censoring. 
Here, $n$, $r$ and $\tau$ are same as defined before, and $\zeta$ is the threshold point based on which we take a decision on the batch. Under such censoring schemes, the proposed estimator of $\lambda$ always exists unlike the MLE, which may not always exist. A loss function, which includes the sampling cost, the cost per unit time, the salvage value and the cost due to acceptance of the batch, is used to determine the DSP, by minimizing the Bayes risk. The optimum DSP is obtained for Type-I and  Type-I hybrid censoring and numerically it has been observed that it is as good as the BSP in terms of the Bayes risk. It is also 
observed that the optimum DSP is better than the sampling plan of \ \cite{Lam1994bayesian} and \cite{lin2008-10exact}. Theoretically 
it has been shown that the implementation of the DSP is easier compared to the BSP proposed by \cite{lin2002bayesian} and \cite{liang2013optimal}, for higher degree polynomial and non-polynomial loss functions.
\par
 The rest of the paper is organized as follows. In Section \ref{S_2}, we present the decision function based on an estimator of 
$\lambda$. All necessary theoretical results for Type-I and  Type-I hybrid censoring are provided in Sections \ref{S_3} and \ref{S_4}, respectively. The DSP for higher degree polynomial and for non-polynomial loss functions are presented in Section \ref{S_6}.  Numerical results are provided in Section \ref{S_5}. Finally, we conclude the paper in Section \ref{S_7}. All derivations are provided in the Appendix.
\section{\sc Problem Formulation and the Proposed Decision Rule}
\label{S_2}
Suppose we are given a batch of items and we need to decide whether we want to accept or reject the batch.  It is assumed that lifetimes of these items are mutually independent and follow an exponential distribution with the probability density function (PDF) 
\begin{equation*}
f(x)=\lambda e^{-\lambda x}; \ \ \ \  x>0 ,\ \ \lambda >0.
\end{equation*}
 To conduct a life testing experiment, $n$ identical items are sampled from the batch and placed on test without replacement with a suitable sampling scheme. Under Type-I and  Type-I hybrid censoring schemes, let $\tau^{*}$ denote the duration of the experiment. Then $\tau^{*}=\tau$ in Type-I censoring and $\tau^{*}=min\{X_{(r)},\tau\}$ in  Type-I hybrid censoring.  Note that $\tau^{*}$ is fixed in Type-I censoring and random for Type-I hybrid censoring.  Let $M$ be the number of failures observed before the fixed time $\tau$,  i.e., $ M = max\{i: X_{(i)}\leq \tau \}$. Hence, the observed sample is $(X_{(1)},X_{(2)},\ldots,X_{(M)})$ in Type-I censoring, and $(X_{(1)},X_{(2)},\ldots,X_{(r)})$ or $(X_{(1)},X_{(2)},\ldots,X_{(M)})$ in  Type-I hybrid censoring. Based on the observed sample, we define the decision function as:
\begin{equation}
\delta(\textbf{x}) =
\begin{cases}
d_{0}         & \mbox{if}  \  \widehat{\lambda} < \zeta, \\
d_{1}         & \mbox{if}   \  \widehat{\lambda} \geq  \zeta,
\end{cases}    \label{dec-func}
\end{equation}
where $\widehat{\lambda}$ is a suitable estimator of $\lambda$,  $\zeta>0$ denotes the threshold point based on which we take a decision on the batch whether to accept (action $d_{0}$) or to reject it (action $d_{1}$).
\par
 Next, we consider a loss function which depends upon various costs.   $C_{r}$ is the cost due to rejecting the batch;  $C_s$ is the cost due to per item inspection; $C_{\tau}$ is the cost of per unit time and $g(\lambda)$ is the cost of accepting the batch. If an item does not fail, then the item can be reused with the salvage value $r_{s}$. Combining all these costs, the general form of the loss function (see \cite{liang2013optimal}, \cite{yang2015optimal}) is given as:
\begin{equation}
L(\delta(\textbf{x}),\lambda,n,r,\tau)=
\begin{cases}
nC_s -(n-M)r_{s}+ \tau^{*} C_{\tau} + g(\lambda) &\mbox{if} \ \delta(\textbf{x}) = d_{0}, \\
nC_s -(n-M)r_{s}+ \tau^{*} C_{\tau} + C_r         &\mbox{if} \ \delta(\textbf{x}) = d_{1}.
\end{cases}   \label{loss-func}
\end{equation}
Clearly $C_s$, $C_{\tau}$, $C_r$ and $r_s$ are non-negative where $C_s>r_s$ and $g(\lambda)$  depends on the  parameter $\lambda$. Smaller $\lambda$ indicates better quality of the item. Therefore,  $g(\lambda)$ can take various forms which have to be positive and increasing with $\lambda$. A quadratic loss function has been widely used as an approximation of the true cost function when a batch is accepted (see \cite{yeh1990optimal}, \cite{Lam1994bayesian} and \cite{yeh1995bayesian}). For a better approximation of the true loss function, higher degree polynomial loss function can be considered, i.e., cost of acceptance in the loss function (\ref{loss-func}) is considered as $g(\lambda)=a_0+a_1\lambda+\ldots+a_k\lambda^k$. It is notable that the true form of the loss function can vary because it includes costs that are difficult to recognize. To obtain the Bayes risk of the decision function (\ref{dec-func}) based on the loss function 
(\ref{loss-func}), it is assumed that $\lambda$ follows a gamma $(a,b)$ prior with the following PDF;
\begin{equation}
\pi(\lambda; a,b)=\frac{b^a}{\Gamma(a)}\lambda^{a-1} e^{-\lambda b}, \ \ \ \  \lambda >0 , \ \ a, b >0. \label{prior}
\end{equation}
Next, we determine the optimum DSP $(n_0,\tau_0,\zeta_0)$ for Type-I censoring and $(n_0,r_0,\tau_0,\zeta_0)$ for  Type-I hybrid censoring such that it has the minimum Bayes risk among all possible sampling plans.

\section{\sc Bayes Risk and DSP under Type -I Censoring}
\label{S_3}

\cite{lin2002bayesian} derived the Bayes risk of the BSP for a quadratic loss function assuming $r_s=0$ and $g(\lambda)=a_0+a_1\lambda+a_2\lambda^2$, such that $a_0 > 0$, $a_1 > 0$ and $a_2 > 0$. This form of the loss function is widely used in the literature (see, for example, \cite{hald}; \cite{Lam1994bayesian}; \cite{yeh1995bayesian}). Likewise, we also derive the Bayes risk of the proposed DSP 
for a quadratic loss function with $r_s>0$, i.e., the loss function takes the following form;
\begin{equation}
L(\delta(\textbf{x}),\lambda,n,\tau)=
\begin{cases}
nC_s - (n-M)r_s +  \tau C_{\tau} + a_0 +a_1\lambda +a_2\lambda^2  &\mbox{if} \quad \delta(\textbf{x}) = d_{0}, \\
nC_s - (n-M)r_s + \tau C_{\tau} + C_r         &\mbox{if} \quad \delta(\textbf{x}) = d_{1}.
\end{cases}   \label{new-loss}
\end{equation}
To derive the Bayes risk for the decision function  (\ref{dec-func}),  we define a suitable estimator of  $\lambda$ as follows:
\begin{equation*}
\widehat{\lambda} =
\begin{cases}
 \ \ \ \ \ \  0 \ \ \ \ \ \ \ \ \ \ \ \ \ \    \mbox{if} \ \ \ M=0 \\
 \ \ \ \ \ \ \widehat{\lambda}_M   \ \ \ \ \ \ \ \ \ \ \  \mbox{if} \ \ \  M >0,
\end{cases}
\end{equation*}
where $\widehat{\lambda}_M$ is the MLE of $\lambda$ given by,
\begin{equation*}
\widehat{\lambda}_M =
 \frac{M}{\sum_{i=1}^{M} X_{(i)} +(n-M){\tau}} \ \ \ \ \mbox{if} \ \ \  M >0.\\
\end{equation*}
\noindent Then the Bayes risk is,
\bea
 r(n, \tau, \zeta) & = & E\big\{L(\delta(\textbf{x}),\lambda, n, \tau)\big\} \nonumber \\
            & = &   E_{\lambda}E_{X/\lambda}\big\{L(\delta(\textbf{x}),\lambda, n, \tau)\big\} \nonumber \\
            & = & n(C_s-r_s)+ E(M)r_{s} + \tau C_{\tau} + a_0 + a_1 \mu_1 + a_2 \mu_2 \nonumber \\
            &  & \quad + E_{\lambda}\big\{(C_r - a_0 - a_1 \lambda - a_2 \lambda^2 )P(\widehat{\lambda} \geq \zeta)\big\},\nonumber 
\eea
where $\ds \mu_{i}= E(\lambda^i)$ for $i=1,2$. Now to find an explicit form of the Bayes risk we need to compute $P(\widehat{\lambda} \geq \zeta)$. Note that the distribution function of $\widehat{\lambda}$ can be written as follows
\begin{align}
\label{mix}
P(\widehat{\lambda} \leq x) &= P(M=0)P(\widehat{\lambda}\leq x|M=0)+P(M \geq 1)P(\widehat{\lambda}\leq x|M \geq 1) \nonumber\\
&= p S(x)+ (1-p)H(x),
\end{align}
where $p=P(M=0)= e^{-n\lambda \tau}$ and 
\begin{align*}
S(x)  =  P(\widehat{\lambda}\leq x|M=0)=
\begin{cases}
1  &\mbox{if} \ \ \  x \geq 0, \\
0   &\mbox{if} \ \ \  otherwise,
\end{cases}\\
H(x)  =  P(\widehat{\lambda}\leq x|M \geq 1)=
\begin{cases}
\int_{0}^{x}h(u)du  &\mbox{if} \ \ \  \frac{1}{n\tau}< x < \infty, \\
0   &\mbox{if} \ \ \  otherwise,
\end{cases}
\end{align*}
where $h(u)$ is the PDF of the absolutely continuous part of the CDF of $\widehat{\lambda}$, and it is provided below.

\begin{lemma}
\label{L_3.1}
The  PDF $h(y)$ is given as
\begin{equation*}
h(y) = 
\frac{1}{1-p}\sum_{m=1}^{n}\sum_{j=0}^{m}\binom{n}{m}\binom{m}{j}(-1)^{j}\frac{e^{ -\lambda (n-m+j)\tau}}{y^2} \pi 
\left (\frac{1}{y}-\tau_{j,m} ; m ,m\lambda \right )
\end{equation*}
for\  $\ds \frac{1}{n\tau}< y < \infty$, \ \ $\ds \tau_{j,m} = (n-m+j) \frac{\tau}{m}$, \  and \ $\pi(\cdot)$ is as defined in 
(\ref{prior}).
\end{lemma}
\begin{proof}
 For $M\geq 1$, The MLE of the mean lifetime $\theta$ is given by $\widehat{\theta}_{M} = \frac{1}{\widehat{\lambda}}$, whose PDF is obtained by \cite{bartholomew1963sampling}. Therefore, the  PDF of $\widehat{\lambda}$, given $M\geq 1$, is obtained by taking the transformation $\widehat{\lambda} =\frac{1}{\widehat{\theta}_M}$, because $\widehat{\theta}_M>0$.
\end{proof}
\noindent Lemma \ref{L_3.1} is used to compute $P(\widehat{\lambda} \geq \zeta)$ and then we used this probability to derive an 
explicit expression of the Bayes risk of the DSP. The following theorem provides the Bayes risk of the DSP for a quadratic loss function (\ref{new-loss}), for any sampling plan $(n,\tau,\zeta)$ .
\newline
\begin{theorem}
\label{T_3.1}
The Bayes risk for the quadratic loss function (\ref{new-loss}) is given by,
\begin{align}
r(n,\tau,\zeta) &= n(C_{s}-r_{s}) + E(M)r_{s} +\tau C_{\tau}+a_0 +a_1 \mu_1 +a_2 \mu_2 
 \nonumber \\            
& + \sum_{l=0}^2 C_{l}\frac{b^{a}}{\Gamma (a)}\bigg[\frac{\Gamma{(a+l)}}{(b+n\tau)^{(a+l)}} \textit{I}_{(\zeta=0)}+\sum_{m=1}^{n}\sum_{j=0}^m (-1)^{j}\binom{n}{m}\binom{m}{j} \frac{\Gamma{(a+l)}}{(C_{j,m})^{a+l}}I_{S^*_{j,m}}(m,a+l)\bigg]. \nonumber
\end{align}
\end{theorem}
\begin{proof}
See Appendix.
\end{proof}
\noindent Since the expression of the Bayes risk $r(n,\tau,\zeta)$ of the DSP is quite complicated, therefore, the optimal values of $n$, $\tau$ and $\zeta$ cannot be computed analytically. \cite{Lam1994bayesian} has given a discretization method to find an optimal sampling plan. Here we use a similar approach to obtain optimal values of $n$, $\tau$ and $\zeta$,  which minimizes the Bayes risk among all sampling plans.\\
 \noindent \textit{Algorithm for finding the optimum DSP:}
\begin{enumerate}
    \item Fix $n$ and $\tau$; minimize $r(n, \tau, \zeta)$ with respect to $\zeta$ using a grid search method  and denote the minimum Bayes risk 
     by $r(n,\tau,\zeta_0 (n,\tau))$.
    \item For fixed $n$, minimize  $ r(n,\tau,\zeta_0 (n,\tau))$ with respect to $\tau$ using a grid search method and denote the minimum Bayes risk by $r(n, \tau_0 (n),\zeta_0 (n,\tau_0 (n)))$.
    \item Choose the sample size $n_{0}$ such that  $$\ds r(n_0, \tau_0 (n_0), \zeta_0 (n_0,\tau_0 (n_0))) \leq r(n, \tau_0 (n),\zeta_0 (n,\tau_0 (n))) \  \ \forall \ \ \ n \geq 0. $$
\end{enumerate}
We denote the optimum DSP by $(n_0,\tau_0,\zeta_0)$ and the corresponding Bayes risk by $r(n_0,\tau_0,\zeta_0)$.  It is observed that the
optimum DSP is unique, see Section \ref{S_5}. The next theorem proves that the proposed algorithm is finite , i.e., we can find an optimum DSP in a finite number of search steps.
\begin{theorem}
\label{T_3.3}
 Assuming $0<\zeta \leq \zeta^{*}$, let us denote $r(n, \tau, \zeta')
= min_{0< \zeta \leq \zeta^{*}} r(n, \tau, \zeta)$ for some fixed $n \ (\geq 1)$ and $\tau$.  Let $n_0$ and $\tau_0$ be the optimal sample 
size and censoring time, respectively . Then,
\begin{align*}
&n_0 \leq \min \bigg\{ \frac{C_r}{C_s-r_s},\frac{a_0 + a_1\mu_1 + \ldots + a_k \mu_k}{C_s-r_s},\frac{r(n, \tau, \zeta')}{C_s-r_s}\bigg\}, \\
&\tau_0 \leq \min \bigg\{ \frac{C_r}{C_{\tau}},\frac{a_0 + a_1\mu_1 + \ldots + a_k \mu_k}{C_{\tau}},\frac{r(n, \tau, \zeta')}{C_{\tau}}\bigg\}.
\end{align*}
\end{theorem}
\begin{proof}
See Appendix.
\end{proof}
 
\section{\sc Bayes Risk and DSP under  Type-I Hybrid Censoring}
\label{S_4}
For  the Type-I hybrid censored sample, \cite{liang2013optimal} derived the Bayes risk of the BSP for a quadratic loss function assuming $g(\lambda)=a_0+a_1\lambda+a_2\lambda^2$, such that $a_0 > 0$, $a_1 > 0$ and $a_2 > 0$. \cite{Chenchou2004} also used the same form 
of the loss function to derive acceptance sampling plans for Type-I hybrid censoring.  Similarly, we also derive the Bayes risk of 
the DSP for a quadratic loss function as follows:
\begin{equation}
L(\delta(\textbf{x}),\lambda,n,\tau)=
\begin{cases}
nC_s - (n-M)r_s +  \tau^{*} C_{\tau} + a_0 +a_1\lambda +a_2\lambda^2  &\mbox{if} \quad \delta(\textbf{x}) = d_{0}, \\
nC_s - (n-M)r_s + \tau^{*} C_{\tau} + C_r         &\mbox{if} \quad \delta(\textbf{x}) = d_{1}.
\end{cases}   \label{new-loss1}
\end{equation}
To derive the Bayes risk for the decision function  (\ref{dec-func})  we define a suitable estimator of $\lambda$ under  Type-I hybrid censoring as:
\begin{equation*}
\widehat{\lambda} =
\begin{cases}
 0 \ \ \ \ \ \ \ \ \ \ \ \ \ \    \mbox{if} \ \ \ M=0 \\
 \widehat{\lambda}_M   \ \ \ \ \ \ \ \ \ \ \  \mbox{if} \ \ \  M >0,
\end{cases}
\end{equation*}
where $\widehat{\lambda}_M$ is the MLE of $\lambda$ given by,
\begin{equation*}
\widehat{\lambda}_M =
\begin{cases}
\frac{M}{\sum_{i=1}^{M} X_{(i)} +(n-M){\tau}} \ \ \ \ \ \mbox{if} \ \ \  M = 1,2 \ldots r-1,  \\
 \frac{r}{\sum_{i=1}^{r} X_{(i)} +(n-r){X_{(r)}}} \ \ \  \mbox{if} \ \ \  M = r.
\end{cases}
\end{equation*}

\noindent Then the Bayes risk is,
\bea
 r(n, r, \tau, \zeta) & = & E\big\{L(\delta(\textbf{x}),\lambda, n, \tau)\big\} \nonumber \\
            & = &   E_{\lambda}E_{X/\lambda}\big\{L(\delta(\textbf{x}),\lambda, n, \tau)\big\} \nonumber \\
            & = & n(C_s-r_s)+ E(M)r_{s} + E(\tau^{*})C_{\tau} + a_0 + a_1 \mu_1 + a_2 \mu_2 \nonumber \\
            &  & \quad + E_{\lambda}\big\{(C_r - a_0 - a_1 \lambda - a_2 \lambda^2 )P(\widehat{\lambda} \geq \zeta)\big\},\nonumber 
\eea
where $\ds \mu_{i}= E(\lambda^i)$ for $i=1,2$. In order to derive an explicit expression of the Bayes risk of the DSP $(n, r, \tau, \zeta)$, we need to compute $P(\widehat{\lambda} \geq \zeta)$.  The distribution of $\widehat{\lambda}$  can be written in a similar form as 
in (\ref{mix}), and the corresponding $h(u)$ is given below.  
\begin{lemma}
\label{L_4.1}
The PDF $h(y)$ is given by 
\begin{align*}
h(y) = 
\frac{1}{1-p}\Bigg[\sum_{m=1}^{r-1}\sum_{j=0}^{m}A_{j,m} &\frac{1}{y^2}\pi 
\left(\frac{1}{y}-\tau_{j,m} ; m ,m\lambda \right) + \frac{1}{y^2}\pi\left(\frac{1}{y}; r,r\lambda\right) \\
 &+r\binom{n}{r}\sum_{k=1}^{r}\binom{r-1}{k-1}(-1)^k \frac{e^{-\lambda(n-r+k)\tau}}{y^2 (n-r+k)}\pi\left(\frac{1}{y}-\tau_{k,r}; r,r\lambda
\right)\Bigg]
\end{align*}
for $\ds \frac{1}{n\tau}< y <\infty$, \ \ $\ds \tau_{j,m} = (n-m+j) \frac{\tau}{m}$,\ \ $\ds A_{j,m}=\binom{n}{m}\binom{m}{j}(-1)^{j}e^{ -\lambda (n-m+j)\tau}$ and $\pi(\cdot)$ as defined in 
(\ref{prior}).
\end{lemma}
\begin{proof}
For $M\geq 1$, The MLE of the mean lifetime $\theta$ is given by $\widehat{\theta}_{M} = \frac{1}{\widehat{\lambda}}$, whose PDF is obtained by \cite{childs2003exact}.  Hence, the PDF of $\widehat{\lambda}$ for $M\geq 1$ can be easily obtained.
\end{proof}

\begin{theorem}
\label{T_4.1}
The Bayes risk using the quadratic loss function (\ref{new-loss1}) is given by 
\bea
 r(n, r,\tau, \zeta) & = & n(C_{s}-r_{s}) + E(M)r_{s} + E(\tau^{*}) C_{\tau}+a_0 +a_1 \mu_1 +a_2 \mu_2 \nonumber\\
            &  & \quad + \sum_{l=0}^{2}C_{l}\frac{b^a}{\Gamma (a)}\bigg\{\frac{\Gamma{(a+l)}}{(b+n\tau)^{(a+l)}} \textit{I}_{(\zeta=0)}+\sum_{m=1}^{n}\sum_{j=0}^{m}\binom{n}{m}\binom{m}{j}(-1)^j R_{l,j,m} +R_{l,r-n,r}\nonumber \\
            &  & \hspace{2.5cm}+\sum_{k=1}^{r}\binom{n}{r}\binom{r-1}{k-1}(-1)^k \frac{r}{(n-r+k)} R_{l,k,r}\bigg\}, \nonumber
\eea
\end{theorem}
\begin{proof}
See Appendix.
\end{proof}
\noindent As in the case of Type-I censoring, the expression of the Bayes risk $r(n,r,\tau,\zeta)$ of the DSP is also quite complicated,
and optimal values of $n$, $r$, $\tau$ and $\zeta$ cannot be computed analytically. In the following steps, an alternative algorithm (see \cite{Lam1994bayesian}) is considered to obtain optimal values of $n$, $r$, $\tau$ and $\zeta$ which minimize the Bayes risk among all sampling plans.
\\
\\
 \noindent \textit{Algorithm for finding optimum DSP:}
\\
To find the optimal values of $n$, $r$, $\tau$ and $\zeta$ based on the Bayes risk, a simple algorithm is described in the following steps:
\begin{enumerate}
    \item Fix $n$, $r$ and $\tau$; minimize $r(n, r,\tau, \zeta)$ with respect to $\zeta$ using a grid search method  and denote the minimum Bayes risk 
     by $r(n,r,\tau,\zeta_0 (n,r,\tau))$.
    \item For fixed $n$ and $r$, minimize  $ r(n,r,\tau,\xi_0 (n, r,\tau))$ with respect to $\tau$ using a grid search method and denote the minimum Bayes risk by $r(n, r, \tau_0 (n,r),\zeta(n,r,\tau_0(n,r)))$.
    \item For fixed $n$, choose $r\leq n$ for which $r(n, r, \tau_0 (n,r),\zeta(n,r,\tau_0(n,r)))$ is minimum and denote it by $r(n, r_0(n), \tau_0 (n,r_0(n)),\zeta(n,r_0(n_0),\tau_0(n,r_0(n))))$.
    
    \item Choose the sample size $n_{0}$ such that  
    \begin{align*}
        r(n_0,r_0(n_0), \tau_0 (n_0,r_0(n_0)), &\zeta_0 (n_0,r_0(n_0),\tau_0 (n_0,r_0(n_0)))) \\
        &\leq r(n, r_0(n), \tau_0 (n,r_0(n)),\zeta_0 (n,r_0(n),\tau_0 (n,r_0(n)))) \  \ \forall \ \ \ n \geq 0.
    \end{align*}
\end{enumerate}
We denote the optimum DSP by $(n_0, r_0,\tau_0,\zeta_0)$ and the minimum Bayes risk by $r(n_0, r_0,\tau_0,\zeta_0)$. In this case  it is also observed that the optimum DSP is unique, see in Section \ref{S_5}.
\par
 It is difficult to find the optimal $\tau_0$ analytically because the Bayes risk expression is complicated. \cite{Tsai2014} suggested a numerical approach to choose a suitable range of $\tau$, say $[0, \tau_\alpha]$ where $\tau_\alpha$ is such that $P(0<X<\tau_{\alpha})=1-\alpha$ and $\alpha$ is a preassigned number satisfying $0<\alpha<1$. The choice of $\alpha$ depends on the prescribed precision. The higher the precision required, the smaller the value of $\alpha$ should be. They suggested the value of $\alpha=0.01$. The next theorem establishes that the proposed algorithm stops in a finite number of steps.
\begin{theorem}
\label{T_4.3}
 Assuming $0<\zeta \leq \zeta^{*}$, let us denote $r(n, r,\tau, \zeta')
= min_{ 0<\zeta \leq \zeta^{*}} r(n,r,\tau, \zeta)$ for some fixed $n \ (\geq 1)$ and $\tau$. Let $n_0$ be the optimal sample 
size. Then,
\begin{align*}
&n_0 \leq \min \bigg\{ \frac{C_r}{C_s-r_s},\frac{a_0 + a_1\mu_1 + \ldots + a_k \mu_k}{C_s-r_s},\frac{r(n, r,\tau, \zeta')}{C_s-r_s}\bigg\}
\end{align*}
and $r_{0}\leq n_{0}$.
\end{theorem}
\begin{proof}
Proof is similar to Theorem \ref{T_3.3}.
\end{proof}

\section{\sc Higher Degree Polynomial and Non Polynomial Loss Functions}
\label{S_6}
In this section, we establish that for a higher degree polynomial loss function or for a non polynomial loss function, the 
implementation of the proposed DSP is much easier compared to the BSP. 

\subsection{\sc Higher Degree Polynomial Loss Function}

In Section \ref{S_3} and \ref{S_4} we consider the quadratic loss function as an approximation of the true loss function. In this section we consider a higher degree polynomial loss function, i.e., the cost of acceptance in the loss function (\ref{loss-func}) is 
$g(\lambda)=a_0+a_1\lambda+\ldots+a_k\lambda^k$.  Based on the discussions in Section \ref{S_4}, it is observed that for $k\geq 5$, the implementation of the proposed DSP under the  Type-I hybrid censoring is straightforward as compared to the BSP.  The Bayes risk for the DSP under  Type-I hybrid censoring for a $k^{th}$ degree polynomial loss function is given by 
\bea
r(n,r,\tau,\zeta) &= & n(C_{s}-r_{s}) + E(M)r_{s} +E(\tau^*) C_{\tau}+a_0 +a_1 \mu_1 +a_2 \mu_2 + \ldots + a_k \mu_k
\nonumber\\
            &  & \quad + \sum_{l=0}^{k}C_{l}\frac{b^a}{\Gamma (a)}\bigg\{\frac{\Gamma{(a+l)}}{(b+n\tau)^{(a+l)}} \textit{I}_{(\zeta=0)}+\sum_{m=1}^{n}\sum_{j=0}^{m}\binom{n}{m}\binom{m}{j}(-1)^j R_{l,j,m} +R_{l,r-n,r}\nonumber \\
            &  & \hspace{2.5cm}+\sum_{k=1}^{r}\binom{n}{r}\binom{r-1}{k-1}(-1)^k \frac{r}{(n-r+k)} R_{l,k,r}\bigg\},
\eea
\noindent where $ E(M)$ and $E(\tau^*)$ are defined earlier. Thus, for any value of $k$, obtaining the Bayes risk is straightforward and the form of the decision function is same for any value of $k$.
\par
\noindent Now in case of the BSP, the Bayes decision rule (see \cite{liang2013optimal}) is given by 
\begin{equation*}
\delta_B(\textbf{x}) =
\begin{cases}
1, & \mbox{if} \ \ \phi_{\pi}\big(m,z\big) \leq C_{r}\\ 
0, &  \ \  \ \  \ \ otherwise,
\end{cases}
\end{equation*}
where, for Type-I censoring
$$
z=\sum_{i=1}^{m} x_{i} + (n-m)\tau,
$$
and for Type-I hybrid censoring
$$
z =
\begin{cases}
\sum_{i=1}^{m} x_{i} + (n-m)\tau & \mbox{if} \ \ \ \ m=1,2 \ldots r-1 \\
\sum_{i=1}^{r} x_{i} + (n-r)x_{r} & \mbox{if}\ \ \ \  m=r, \\
\end{cases}
$$
with
$$
\phi_{\pi}\big(m,z\big) =\int_{0}^{\infty}g(\lambda)\pi\big(\lambda|m,z\big)d\lambda.
$$

\noindent Since the prior distribution of $\lambda$ is gamma $(a,b)$, it is well known that the posterior distribution of $\lambda$ is also gamma, viz.,  
\begin{equation*}
\pi\big(\lambda|m,z \big) \sim \hbox{gamma} (m+a,z+b).
\end{equation*}
Now when $g(\lambda)= a_0+a_1\lambda+\ldots+a_k\lambda^k$ in ($\ref{loss-func}$) then, 
\begin{align*}
    \phi_{\pi}\big(m, z\big) = \int_{0}^{\infty}g(\lambda)\pi\big(\lambda|m,z\big)d\lambda = a_{0} +\sum_{j=1}^{k}a_{j}\frac{(m+a)\ldots(m+a+j-1)}{(z+b)^j}.
\end{align*}
Thus, to find the closed form of the decision function we need to obtain the set
$$
A = \{z;\ z \ge 0, \phi_{\pi}\big(m, z \big) \leq C_{r}\},
$$
and to construct the set $A$, we need to obtain the set of $z \ge 0$, such that   
\begin{equation}
h_1(z) = a_{0} +\sum_{j=1}^{k}a_{j}\frac{(m+a)\ldots(m+a+j-1)}{(z+b)^j} \leq C_{r},   \label{hdp1}
\end{equation}
which is equivalent to find $z \ge 0$, such that,
\begin{eqnarray}
h_2(z) & = & (C_{r}-a_{0}) \big(z+b\big)^k-\sum_{j=1}^{k}a_{j}(m+a)\ldots(m+a+j-1)\big(z+b\big)^{k-j}  \geq 0.   \label{hdp-21}
\end{eqnarray}
It can be easily shown that if $D_n(m)$ is the only real root or $D_n(m)$ is the maximum real root of $h_2(z) = 0$ then the Bayes decision function will take the following form.
\begin{equation}
\delta_B(\textbf{x}) =
\begin{cases}
1, & \mbox{if} \ \  z \geq a(n,r,\tau,m)\\ 
0, &  \ \  \ \  \ \ otherwise,
\end{cases}
\label{hdp-df1}
\end{equation}
where \ 
$a(n,r,\tau,0)=0\vee(D_n(0)-b)$ and  $a(n,r,\tau,m)=0\vee(D_n(m)-b)\wedge n\tau \ \ \forall \  1\leq m\leq r$.
However, it is not straightforward to find the real root when $k\geq5$. It is well known that there is no algebraic solution to polynomial equations of degree five or higher (see chapter 5,  \cite{Herstein1975}). So the BSP cannot be obtained for fifth or higher degree polynomial loss function analytically.  Even, finding the optimal sampling plan numerically becomes very difficult.

\subsection{\sc Non-Polynomial Loss Function}
We have already discussed in Section \ref{S_2} that the loss due to accepting the batch $g(\lambda)$ can vary and the true form of the loss function is likely to be unknown.  When we have a non-polynomial loss function, we show that implementation of the proposed DSP is quite easy and the associated Bayes risk is computed without any additional effort as compared to the BSP. To illustrate this, we use the following non polynomial loss function: 
\begin{equation}
L(\delta(\textbf{x}),\lambda,n,r,\tau)=
\begin{cases}
nC_s -(n-M)r_{s}+ \tau^{*} C_{\tau} + a_0+a_1\lambda+a_2\lambda^{5/2} &\mbox{if} \ \delta(\textbf{x}) = d_{0}, \\
nC_s -(n-M)r_{s}+ \tau^{*} C_{\tau} + C_r         &\mbox{if} \ \delta(\textbf{x}) = d_{1},
\end{cases}   \label{nploss-func}
\end{equation}
\noindent where $g(\lambda)= a_0+a_1\lambda+a_2\lambda^{5/2}$ is an increasing function in $\lambda$. Here we consider only  the 
Type-I hybrid censoring case.  The Bayes risk of the DSP for the Type-I hybrid censoring is as follows: 
\bea
r(n,r,\tau,\zeta) &= & n(C_{s}-r_{s}) + E(M)r_{s} +E(\tau^*) C_{\tau}+a_0 +a_1 \mu_1 +a_2 \frac{\Gamma{(a+\frac{5}{2}})}{\Gamma{(a)}b^{\frac{5}{2}}} 
\nonumber\\
            &  & \quad + \sum_{l=0}^{2}C_{l}\frac{b^a}{\Gamma (a)}\bigg\{\frac{\Gamma{(a+p_l)}}{(b+n\tau)^{(a+p_l)}} \textit{I}_{(\zeta=0)}+\sum_{m=1}^{n}\sum_{j=0}^{m}\binom{n}{m}\binom{m}{j}(-1)^j R_{p_l,j,m} \nonumber \\
            &  & \hspace{2.5cm}+ R_{p_l,r-n,r} +\sum_{k=1}^{r}\binom{n}{r}\binom{r-1}{k-1}(-1)^k \frac{r}{(n-r+k)} R_{p_l,k,r}\bigg\}, \nonumber
\eea
\noindent where $E(M)$ and $E(\tau^*)$ are defined earlier and 
\begin{equation*}
p_l=
\begin{cases}
0, & \mbox{if} \ \ \ \ l=0 \\
1, & \mbox{if} \ \ \ \ l=1\\
\frac{5}{2}, & \mbox{if} \ \ \ \ l=2.
\end{cases}
\end{equation*}
To express the Bayes decision function of the BSP (see \cite{liang2013optimal}) in a simpler form 
for the non-polynomial loss function $g(\lambda)=a_0+a_1\lambda+a_2\lambda^{5/2}$, 
we have to consider,
\begin{align*}
    \phi_{\pi}\big(m, y(n,r,\tau,m)\big) & = \int_{0}^{\infty}g(\lambda)\pi\big(\lambda|m,y(n,r,\tau,m)\big)d\lambda\\
    &= a_{0} + \frac{a_{1}(m+a)}{(y(n,r,\tau,m)+b)}+ \frac{a_{2}\Gamma(m+a+\frac{5}{2})}{\Gamma(m+a)(y(n,r,\tau,m)+b)^{\frac{5}{2}}}.
\end{align*}
So to find a closed form of the decision function we need to obtain the set 
$$
A = \{x;\ x \ge 0, \phi_{\pi}\big(m, x \big) \leq C_{r}\}.
$$
Note that to construct the set $A$, we need to obtain the set of $x \ge 0$ such that
\begin{equation}
h_1(x) = a_{0} + \frac{a_{1}(m+a)}{(x+b)}+ \frac{a_{2}\Gamma(m+a+\frac{5}{2})}{\Gamma(m+a)(x+b)^{\frac{5}{2}}} \leq C_{r},   \label{hdnp} \nonumber
\end{equation}
and this is equivalent to find $x \ge 0$ such that  
\begin{eqnarray}
h_2(x) = (C_{r}-a_{0})\Gamma(m+a) \big(x+b\big)^{\frac{5}{2}}-a_{1}(m+a)\Gamma(m+a)\big(x+b\big)^{\frac{3}{2}}-a_{2}\Gamma(m+a+\frac{5}{2})\geq 0.   \label{hdnp-2} \nonumber 
\end{eqnarray}
It is obvious that we cannot obtain a closed form solution of the non polynomial equation $h_2(x)=0$. So in  case of a general 
non-polynomial loss function, we cannot construct a closed form of the Bayes decision function and obtain the explicit expression of Bayes risk. But since our decision function does not depend on the form of the loss function, this difficulty does not arise in case of the
proposed DSP.

\section{\sc Numerical Results and Discussion}
\label{S_5}
To obtain the numerical results, we consider the algorithms proposed in Sections \ref{S_3} and \ref{S_4} for Type-I and  Type-I hybrid censoring, respectively.  Let us assume that $n_{1}^{*}$ and $n_{2}^{*}$ from Theorem \ref{T_3.3} and \ref{T_4.3} denote the upper bound of $n_0$ under Type-I censoring and Type-I hybrid censoring, respectively. Then, for Type-I censoring
$$
 0\leq n_0 \leq n_{1}^{*} \ \ \ and \ \ \  0 \leq \tau_0 \leq \tau', 
$$
 and for Type-I hybrid censoring
$$
 0\leq n_0 \leq n_{2}^{*} \ \ \ and \ \ \  0 \leq \tau_0 \leq \tau_{\alpha},
$$
 where $\tau'$ and $\tau_{\alpha}$ are upper bounds of $\tau$ under Type-I and Type-I hybrid censoring. For fixed $n$ and $\tau$ in Type-I censoring and for fixed $n$, $r \ (\leq n)$  and $\tau$ in Type-I hybrid censoring, we minimize the Bayes risk  with respect to $\zeta$ using a grid search method where the grid size of $\zeta$ is taken as $0.0125$. Then, we minimize with respect to $\tau$ where grid size of $\tau$ is taken as $0.0125$. Finally, we choose the value of $n$ in Type-I censoring and the value of $n$ and $r \ (\leq n)$ in Type-I hybrid censoring for which the Bayes risk is minimum.
 
\subsection{\sc Comparison with Lam (1994), Lin \textit{et al.} (2010) and BSP sampling plans}

 In this section, we focus on comparing the optimum DSP with \cite{Lam1994bayesian}, \cite{lin2010co} and BSP sampling plans. For Type-I censoring, comparison with \cite{Lam1994bayesian} and \cite{lin2010co} sampling plans the values of coefficients $\ a_0=2,\ a_1=2,\ a_2=2,\ C_s=0.5,\ C_{\tau}=0$, $r_{s}=0$ and $C_r=30$ are used. In Table \ref{table-1com} only hyper-parameters $a$ and $b$ are varying and others are kept fixed.
 
\begin{table}[H]
\caption{Numerical comparison with Lam (1994) and Lin \textit{et al.} (2010) sampling plans for different values of $a$ and $b.$}
\label{table-1com}
		\resizebox{17cm}{0.2\textheight}{
		\begin{tabular}{|c|cc|cccc|cc|cccc|}\hline
	Scheme &  $a$ & $b$ & $r( n_0, \tau_0, \zeta_0(\xi_0))$ & $n_0$ & $\tau_0$ & $\zeta_0(\xi_0)$  &  $a$ & $b$ & $r( n_0, \tau_0, \zeta_0)$ & $n_0$ & $\tau_0$ & $\zeta_0(\xi_0)$ \\ \hline
	DSP & 2.5 & 0.8 & 24.8419 & 4 & 1.3125 & 3.0475 & 1.5 & 0.8 & 16.5825 & 3 & 0.7000 & 4.2862 \\ 
	LAM &     &     & 24.9367 & 3 & 0.7077 & 0.3539 &     &     & 16.6233 & 3 & 0.5262 & 0.2631 \\
    Lin et al.(2010) &     &     & 24.9893 & 4 & 0.6808 & 0.3404 &  &  & 16.7533 & 3 & 0.5262 & 0.2631  \\ \hline
	DSP & 2.5 & 1.0 & 21.7081 & 4 & 1.1125 & 3.5950 & 2.0 & 0.8 & 21.1398 & 4 & 1.1625 & 3.4500 \\ 
	LAM  &     &     & 21.7640 & 3 & 0.5483 & 0.2742  &     &     & 21.2153 & 3 & 0.6051 & 0.3026 \\
    Lin et al.(2010)  &     &     & 21.8515 & 4 & 0.5819 & 0.2910 &  &  & 21.2875 & 4 & 0.6051 & 0.3026  \\  \hline
	DSP & 3.0 & 0.8 & 27.5581 & 3 & 1.1625 & 2.5875 & 2.5 & 0.6 & 27.7267 & 3 & 1.2125 &  2.4863 \\ 
	LAM &     &     & 27.6136 & 3 & 0.8170 & 0.4085  &     &     & 27.7834 & 3 & 0.8537 & 0.4268  \\ 
	Lin et al.(2010)  &     &     & 27.6521 & 3 & 0.8170 & 0.4085 &  &  & 29.8193 & 3 & 0.8537 & 0.4268 \\ \hline
	DSP & 3.5 & 0.8 & 29.2789 & 2 & 1.0125 & 1.9875 & 10.0 & 3.0 & 29.5166 & 2 & 0.8000 & 2.5187 \\
	LAM &     &     & 29.2789 & 2 & 1.0037 & 0.5019  &     &     & 29.5166 & 2 & 0.7928 & 0.3964  \\
	Lin et al.(2010)  &     &     & 29.3642 & 2 & 1.0037 & 0.5019 &     &     & 29.5959 & 2 & 0.8194 & 0.4097   \\ \hline
		\end{tabular}
  }
\end{table}
\vspace{-1em}
From Table \ref{table-1com} it is clear that Bayes risk of the optimum DSP is less then or equal to the Bayes risk of \cite{Lam1994bayesian} and \cite{lin2010co} sampling plans.

\noindent For Type-I censoring to compare with the BSP proposed by \cite{lin2002bayesian} we use set of coefficient $ a_0=2, a_1=2, a_2=2, C_s=0.5 , C_\tau=0.5, C_r=30$, the prior parameters $a=2.5, b=0.8$ and assume $\zeta^{*}=6$. We obtain the minimum Bayes risk and decision theoretic sampling plan (DSP) for the proposed method and the Bayes risk of the BSP by varying $a$, $b$, $C_s$, $C_\tau$ and $C_r$ one at a time and keeping other fixed. The results are presented in Table $\ref{T1-HT1}$.
\begin{table}[H]
\caption{Numerical Comparison between DSP and BSP for Type-I censoring and Hybrid Type-I censoring.}
\vspace{0.5em}
\resizebox{17cm}{0.3\textheight}{
		\begin{tabular}{|cc|c|cccc|cc|c|ccccc|}\hline
		\multicolumn{1}{|c}{} &
	     \multicolumn{6}{c|}{Type-I censoring} & \multicolumn{1}{c}{} &  \multicolumn{7}{c|}{Hybrid Type-I censoring } \\ \hline
	   \multicolumn{1}{|c}{$a$} &
		\multicolumn{1}{c|}{$b$} &
		\multicolumn{1}{c}{BSP} &
	    \multicolumn{4}{|c|}{DSP} & 
	    \multicolumn{1}{c}{$a$} &
		\multicolumn{1}{c|}{$b$} &
		\multicolumn{1}{c|}{BSP\footnotemark} &
		\multicolumn{5}{c|}{DSP} \\ \cline{3-7} \cline{10-15}
	   & & $r( n_B, \tau_B, \delta_B)$ & $r( n_0, \tau_0, \zeta_0)$ & $n_0$ & $\tau_0$ & $\zeta_0$ & & & $r( n_B, \tau_B, \delta_B)$ & $r( n_0, r_0, \tau_0, \zeta_0)$ & $n_0$ & $r_0$ & $\tau_0$ & $\zeta_0$  \\ \hline
	   2.5 & 0.8 &25.2777 & 25.2777 & 3 & 0.7250 & 2.9750 &  2.5 & 0.8 & 26.0319 & 26.0338 & 6 & 3 & 0.2000 & 2.9750 \\ \hline
	   	2.5 & 1.0 & 22.0361 & 22.0361 & 3 & 0.5625 & 3.7250 & 2.5 & 1.0 & 22.6430 & 22.6437 & 5 & 3 & 0.1875 & 3.7200 \\ \hline
	   	3.5 & 0.8 & 29.7131 & 29.7131 & 2 & 0.8125 & 1.9875 &  3.0 & 0.8 & 28.7885 & 28.7889 & 4 & 2 & 0.2375 & 2.3445   \\ \hline
	    \multicolumn{2}{|c}{$C_s$} &
	     \multicolumn{5}{|c|}{} & \multicolumn{2}{c}{$C_s$} &  \multicolumn{6}{|c|}{} \\ \hline
		\multicolumn{2}{|c|}{0.50} & 25.2777 & 25.2777 & 3  & 0.7250 & 2.9750  & \multicolumn{2}{c|}{0.30} & 24.3326 & 24.3341 & 10 & 4 & 0.1500 & 3.0500 \\ \hline
	    \multicolumn{2}{|c|}{1.00} & 26.5396 & 26.5396 & 2  & 0.5875 & 2.8625 &  \multicolumn{2}{c|}{0.50} & 26.0319 & 26.0338 & 6  & 3 & 0.2000 & 2.9750 \\ \hline
		\multicolumn{2}{|c|}{2.00} & 27.9542  & 27.9542 & 1  & 0.3750 & 2.6750 & 	\multicolumn{2}{c|}{0.70} &  26.9106 & 26.9114 & 3  & 2 & 0.2750 & 2.8625 \\ \hline
	    \multicolumn{2}{|c}{$C_\tau$} &
	    \multicolumn{5}{|c|}{} & \multicolumn{2}{c}{$C_\tau$} &  \multicolumn{6}{|c|}{} \\ \hline
		\multicolumn{2}{|c|}{0.50} & 25.2777 & 25.2777 & 3  & 0.7250 & 2.9750 & \multicolumn{2}{c|}{0 }& 24.6354 & 24.6754 & 4 & 4 & 0.8750 & 3.0500  \\ \hline
	    \multicolumn{2}{|c|}{1.00} &  25.6238 & 25.6238 & 3  & 0.6625 & 2.9750  & \multicolumn{2}{c|}{8} & 26.4662 & 26.4672 & 7  & 3 & 0.1625 & 2.9750 \\ \hline
		\multicolumn{2}{|c|}{2.00} & 26.1439 &  26.1439 & 4  & 0.3875 & 2.9750 & \multicolumn{2}{c|}{16} & 27.2513 & 27.2513 & 7  & 2 & 0.1000 & 1.9625  \\ \hline
	    \multicolumn{2}{|c}{$C_r$} &
	    \multicolumn{5}{|c|}{} & \multicolumn{2}{c}{$C_r$} &
	    \multicolumn{6}{|c|}{} \\ \hline
		\multicolumn{2}{|c|}{20} & 19.3293 & 19.3293 & 2 & 0.8750 & 1.7750   & \multicolumn{2}{c|}{25} & 23.3583  & 23.3581 & 4  & 2 & 0.2375 & 2.2875   \\ \hline
	    \multicolumn{2}{|c|}{30} & 25.2777 & 25.2777 & 3  & 0.7250 & 2.9750  & \multicolumn{2}{c|}{30} & 26.0319 & 26.0338 & 6  & 3 & 0.2000 & 2.9750  \\ \hline
		\multicolumn{2}{|c|}{50} & 32.2092 & 32.2092 & 5  & 0.5625 & 5.0500 & \multicolumn{2}{c|}{40} & 30.0072 & 30.0069 & 7  & 4 & 0.1750 & 4.0750  \\ \hline
	   	\end{tabular}
}
		 \label{T1-HT1}
\end{table}
\vspace{-1em}
\footnotetext[1]{Bayes risk of BSP is obtained by simulation.}
\noindent
Similarly for the Type-I hybrid censoring, comparison is made between the BSP proposed by Liang and Yang (2013) and the proposed DSP  by taking set of coefficients $ a_0=2, a_1=2, a_2=2, C_s=0.5, r_s=0.3, C_\tau=5.0, C_r=30$, the hyper parameters $ a=2.5, b=0.8 $ and assume $\zeta^{*}=6$. The Bayes risk of the BSP involves a complicated integral, and it has been approximated by Monte Carlo simulation. So 
the Bayes risk here is an approximation of the exact Bayes risk of the BSP. 
\par
Further, when the Bayes risk has a unique minimum, the proposed algorithm gives us the optimum DSP without any additional 
computational effort.  Since the Bayes risk expression is quite complicated, it is not easy to prove theoretically that the function has a unique minimum. So we study graphical behavior of the Bayes risk by providing its contour plots with respect to $\tau$ (on x-axis) and $\zeta$ (on y- axis) with hyper parameter $a=2.5$, $b=0.8$ and set of coefficients mentioned above for Type-I and  Type-I hybrid censoring.

\begin{figure}[H]
\vspace{-0.5em}
\hspace*{-1.5cm} 
\subfloat{
  \includegraphics[width=10cm, height=10cm]{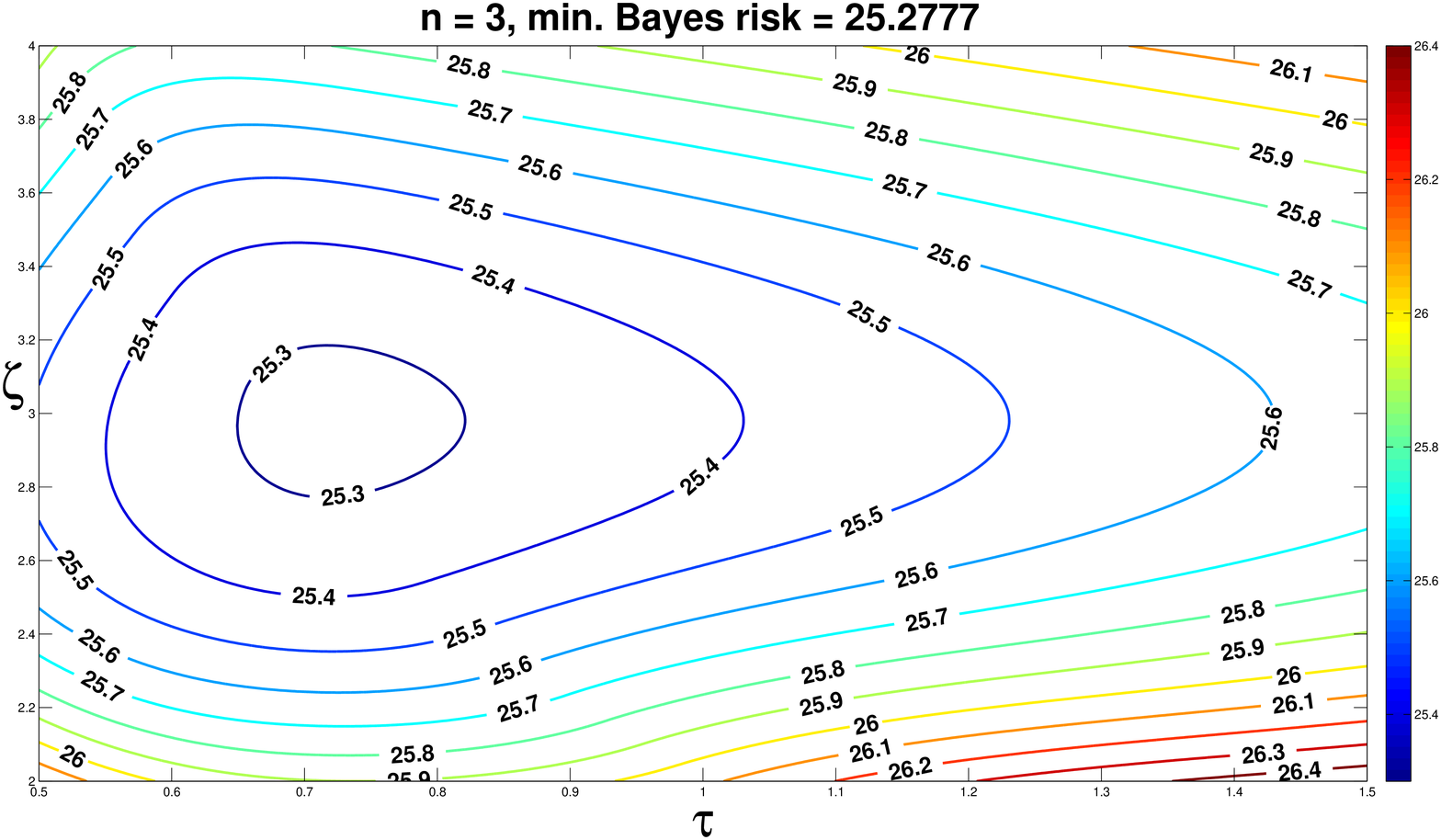}
}
\hspace{-1cm}%
\subfloat{
  \includegraphics[width=10cm, height=10cm]{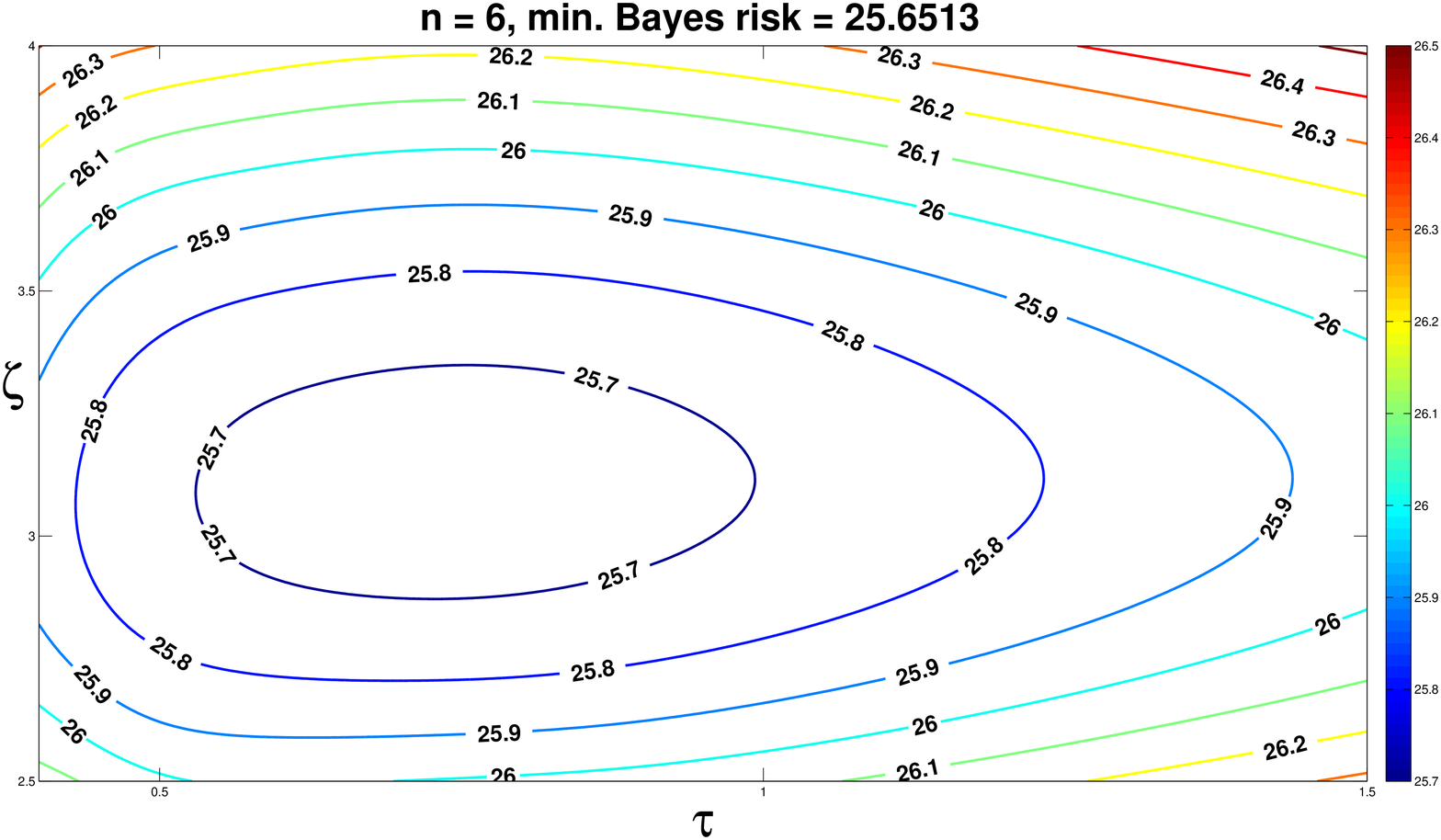}
}
\caption{Contour plot of Bayes risk with set of coefficient $a_0=2, a_1=2, a_2=2, C_s=0.5, r_s=0, C_{\tau}=0.5, C_r=30$ and $a=2.5, b=0.8$ for Type-I censoring.}
\label{fig:fig1}
\vspace{-2em}
\end{figure}

In Type-I censoring, the Bayes risk is a function of three parameters which are $n$, $\tau$ and $\zeta$, among which one is discrete and two are continuous. Since $n$ takes discrete values and from Theorem \ref{T_3.3} we know that optimal value of $n$ is bounded above, so for different values of $n$, we provide the contour plot of Bayes risk with respect to $\tau$ and $\zeta$ in Figure \ref{fig:fig1}. It is clear from contour plot that the Bayes risk has a unique minimum with respect to $\tau$ and $\zeta$. We also observe that the Bayes risk first decreases then increases as $n$ increases.
\begin{figure}[H]
\vspace{-1em}
\hspace*{-1.5cm} 
\subfloat{
  \includegraphics[width=10cm, height=10cm]{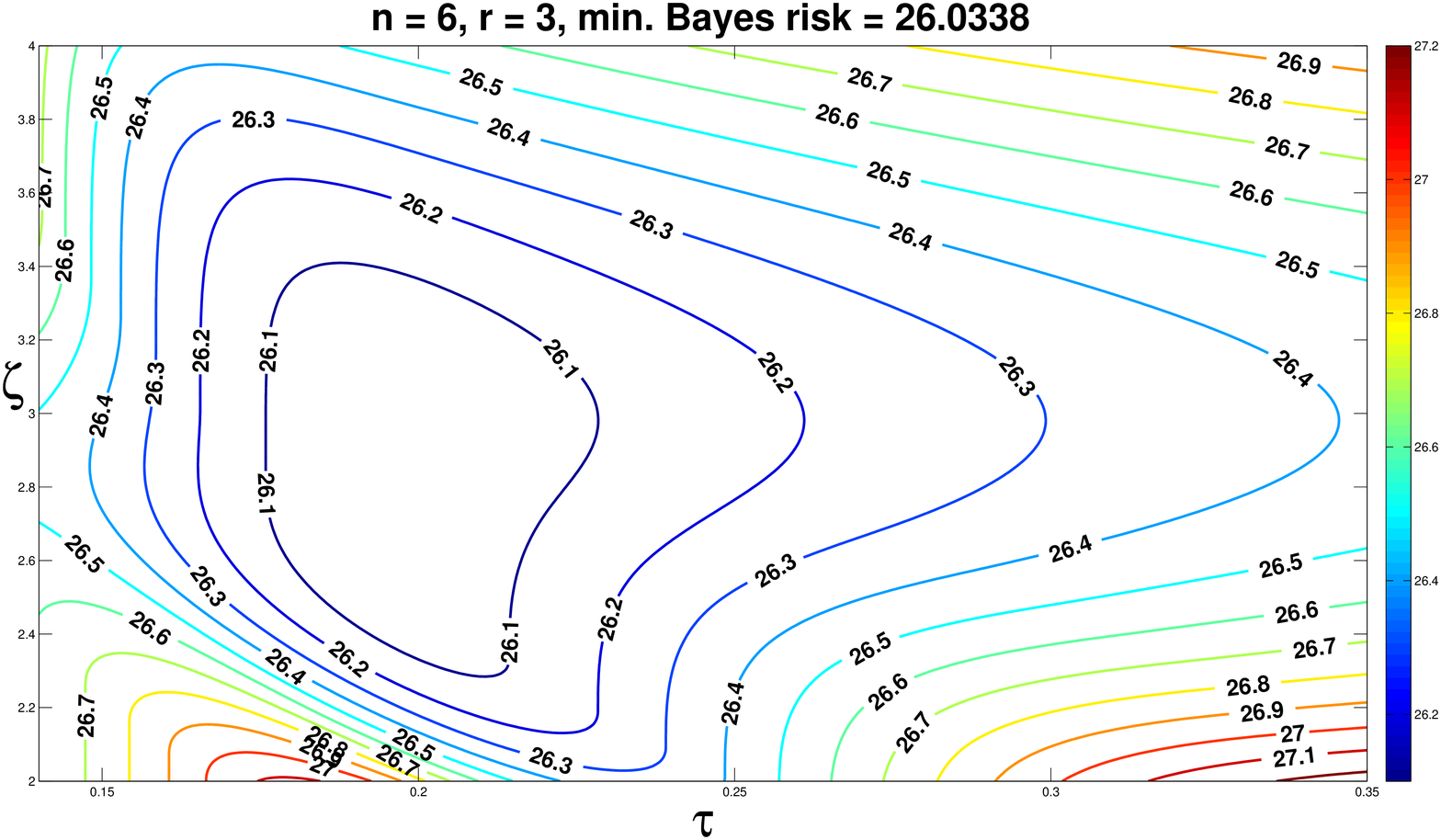}
}
\hspace{-1cm}%
\subfloat{
  \includegraphics[width=10cm, height=10cm]{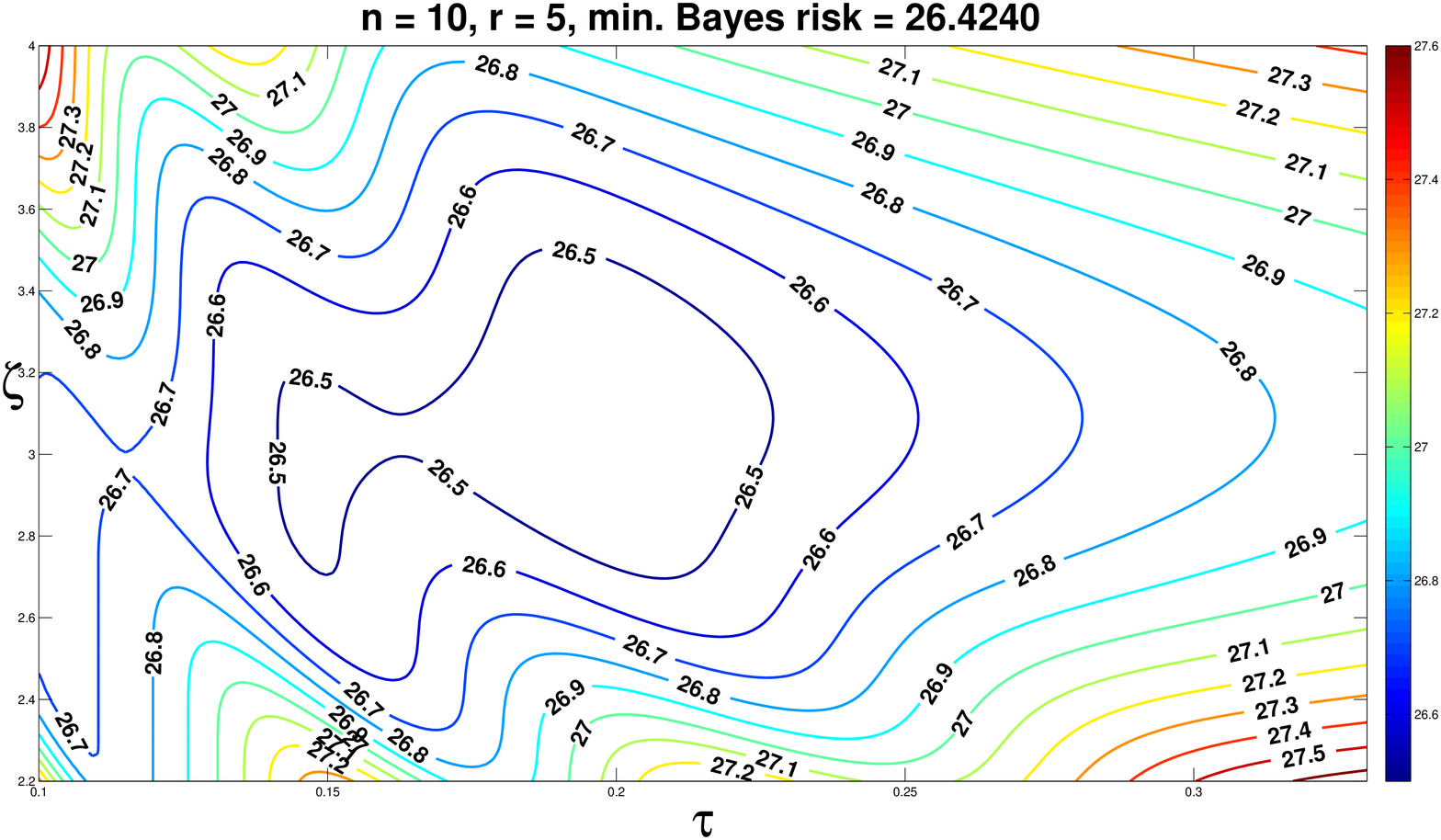}
}
\caption{Contour plot of Bayes risk with set of coefficient $a_0=2, a_1=2, a_2=2, C_s=0.5, r_s=0.3, C_{\tau}=5.0, C_r=30$ and $a=2.5, b=0.8$ for  Type-I hybrid censoring.}
\label{fig:fig2}
\end{figure}
 Similarly for  Type-I hybrid censoring, Bayes risk is a function of four parameters $n$, $r$, $\tau$, and $\zeta$ among which two are discrete and two continuous. Since optimal values of $n$ and $r$ are bounded above(see Theorem \ref{T_4.3}), so for different values of $n$ and $r$, we provide the contour plots of Bayes risk with respect to $\tau$ and $\zeta$ in Figure \ref{fig:fig2}. In this case also Bayes risk has unique minimum and as $n$ increases Bayes risk first decreases then increases. The contour plot can also be used for  predicting the range which includes the optimal values of $\tau$ and $\zeta$.
 
\subsection{\sc Numerical results for Higher degree polynomial and Non polynomial loss function}
  In Section \ref{S_6} we have observed that for a higher degree polynomial and for a non polynomial loss function the DSP can be obtained without any additional effort as compared to the BSP. The numerical results for fifth degree polynomial and for non polynomial loss function are tabulated in Tables \ref{HItable-1}-\ref{NTtable-5}. Standard values of parameter, coefficients and costs are defined in every section where needed.  In each table, only hyper parameters $a$ and $b$ or one coefficient or one cost  can change and the others are kept fixed.  It should be clear from the tables.
\subsubsection{\sc Fifth Degree Polynomial Loss Function}
For Type-I hybrid censoring, we present the optimum DSP for fifth degree polynomial loss function, with the standard set of hyper parameter, coefficients and costs: $a=1.5, b=0.8, a_0 = 2, a_1 = 2, a_2 = 2, a_3=2, a_4=2, a_5=2, C_r = 30, C_s = 0.5, r_s=0.3, C_\tau = 0.5, \zeta^{*}=6$. In Tables \ref{HItable-1}-\ref{HItable-6} the values of the different  hyper parameters or coefficients or costs are given in column $1$ and $7$. The minimum Bayes risk $r(n_0,r_0,\tau_0,\zeta_0)$ and the corresponding optimal sampling plan $(n_0,r_0,\tau_0,\zeta_0)$ are given in columns $2-6$ and $8-12$.

\begin{table}[H]
\caption{The minimum Bayes risk and optimum DSP for Type-I hybrid censoring as $a$ and $b$ varies }
\vspace{1em}
\centering
\resizebox{8.5cm}{0.1\textheight}{
		\begin{tabular}{|cc|ccccc|}\hline
	  $a$ & $b$  & $r( n_0, r_0, \tau_{0},\zeta_0)$ & $n_0$ & $r_0$ &  $\tau_{0}$ & $\zeta_0$  \\ \hline
         0.2 &   0.2 & 12.1795   &  5   & 4  & 1.2625  & 0.9750  \\ \hline
	     1.5 &   0.4 &  29.6469 &    2 &    2 &    2.9125  &  0.6250 \\ \hline
		 1.5 &    0.8 & 26.2983  & 5    & 4    & 1.6375     & 0.9250  \\ \hline
	     2.5 & 1.5 &  27.8324  &  5  &    4   &   1.6750  &  0.9250  \\ \hline
	     3.0 &   1.5 &   29.9061   & 4 &    3   &    1.7000  &  0.7500 	\\ \hline
	   	\end{tabular}
}
		 \label{HItable-1}
\end{table}

\begin{table}[H]
\caption{ The minimum Bayes risk and optimum DSP for Type-I hybrid censoring as $a_0$ or $a_1$ varies}
\vspace{1em}
\resizebox{16.5cm}{0.1\textheight}{
		\begin{tabular}{|c|ccccc|c|ccccc|}\hline
	  $a_0$  & $r( n_0, r_0,\tau_{0},\zeta_0)$ & $n_0$ & $r_0$ & $\tau_{0}$ & $\zeta_0$ & $a_1$ & $r( n_0, r_0, \tau_0, \zeta_0)$ & $n_0$ & $r_0$ & $\tau_{0}$ & $\zeta_0$  \\ \hline
	    0.5 & 25.8251 &   6 &   5 &   1.6625 &    1.0000 & 0.5 &  26.0091 &   6 &    5 &    1.6625  & 1.0000
 \\ \hline
	     1.0 & 25.9891  &  5  &   4   & 1.6250    & 0.9375  & 1.0 &   26.1080 &    5  &    4 &     1.6250  &    0.9375\\ \hline
	   1.5 &  26.1444  & 5  &  4  &    1.6375  &  0.9375 &    1.5 &  26.2038 &   5  &    4 &    1.6375 &   0.9375\\ \hline
	   	 2.0 &  26.2983  &  5    & 4   & 1.6375  &  0.9250
 & 2.0 &   26.2983 &    5 &    4 &    1.6375 &    0.9250\\ \hline
	    2.5 &  26.4515  &  5  &   4  &  1.6500 &   0.9250
 &  2.5 &   26.3919 &   5 &    4 &    1.6500 &    0.9250 \\ \hline
	   	\end{tabular}
}
		 \label{HItable-2}
\end{table}

\begin{table}[H]
\caption{  The minimum Bayes risk and optimum DSP for Type-I hybrid censoring as $a_2$ or $a_3$ varies }
\vspace{1em}
\resizebox{16.5cm}{0.1\textheight}{
		\begin{tabular}{|c|ccccc|c|ccccc|}\hline
	  $a_2$  & $r( n_0, r_0,\tau_{0},\zeta_0)$ & $n_0$ & $r_0$ & $\tau_{0}$ & $\zeta_0$ & $a_3$ & $r( n_0, r_0, \tau_0, \zeta_0)$ & $n_0$ & $r_0$ & $\tau_{0}$ & $\zeta_0$  \\ \hline
	    0.5 &  26.0403 &    6  &   5  &    1.6625 &   1.0125  & 0.5 &   25.9983  &    6  &    5  &    1.6250 &   1.0125 
 \\ \hline
	     1.0 & 26.1284  &  5  &   4  &    1.6250  &    0.9375 & 1.0 &  26.1027  &   5  &    4  &   1.6125 &    0.9500\\ \hline
	   1.5 & 26.2141 &   5 &    4 &    1.6250 &   0.9375 &    1.5 & 26.2022    & 5  & 4   &    1.6250 &   0.9375 \\ \hline
	   	 2.0 & 26.2983  &   5  &    4  &   1.6375  & 0.9250
 & 2.0 &    26.2983 &    5  &    4  &    1.6375 &   0.9250 \\ \hline
	    2.5 &  26.3810 &   5  &    4  &    1.6500 &   0.9250
 &  2.5 &  26.3911 &   5   &    4  &   1.6500  &  0.9125 \\ \hline
	   	\end{tabular}
}
		 \label{HItable-3}
\end{table}

\begin{table}[H]
\caption{ The minimum Bayes risk and optimum DSP for Type-I hybrid censoring as $a_4$ or $a_5$ varies }
\vspace{1em}
\resizebox{16.5cm}{0.1\textheight}{
		\begin{tabular}{|c|ccccc|c|ccccc|}\hline
	  $a_4$  & $r( n_0, r_0,\tau_{0},\zeta_0)$ & $n_0$ & $r_0$ & $\tau_{0}$ & $\zeta_0$ & $a_5$ & $r( n_0, r_0, \tau_0, \zeta_0)$ & $n_0$ & $r_0$ & $\tau_{0}$ & $\zeta_0$  \\ \hline
	    0.5 &  25.8618 &   6 &   5 &    1.5875 &   1.0500 & 0.5 &   25.4497 &  5  & 4   &   1.4125 &   1.0750
 \\ \hline
	     1.0 &  26.0212  &  5 &   4 &    1.5750 & 0.9625 & 1.0 &   25.8046  &  5 &    4 &    1.5000 & 1.0125\\ \hline
	   1.5 & 26.1656  &   5 &     4  &    1.6125  &    0.9500 &    1.5 &  26.0771 &   5   &    4   &    1.5750  &  0.9625\\ \hline
	   	 2.0 & 26.2983  &   5  &    4  &    1.6375  &   0.9250
 & 2.0 &    26.2983 &    5  &    4  &    1.6375 &   0.9250 \\ \hline
	    2.5 &  26.4212  &  5  &    4 &     1.6750 &   0.9125
 &  2.5 & 26.4838 &   5  &    4 & 1.7000   &    0.9000 \\ \hline
	   	\end{tabular}
}
		 \label{HItable-4}
\end{table}

\begin{table}[H]
\caption{ The minimum Bayes risk and optimum DSP for Type-I hybrid censoring as $C_s$ or $C_{\tau}$ varies}
\vspace{1em}
\resizebox{16.5cm}{0.1\textheight}{
		\begin{tabular}{|c|ccccc|c|ccccc|}\hline
	  $C_s$  & $r( n_0, r_0,\tau_{0},\zeta_0)$ & $n_0$ & $r_0$ & $\tau_{0}$ & $\zeta_0$ & $C_{\tau}$ & $r( n_0, r_0, \tau_0, \zeta_0)$ & $n_0$ & $r_0$ & $\tau_{0}$ & $\zeta_0$  \\ \hline
	    0.4 & 25.6655 &   7  &    5  &   1.2375 & 0.9875 &  0.2  &  25.9954 &    4 &    4 &    3.1375 &   0.9250
 \\ \hline
	     0.5 &   26.2983 &    5  &    4 &    1.6375 & 0.9250 &  0.5 &   26.2983 &    5 &    4 &   1.6375 &   0.9250 \\ \hline
	   0.8  &   27.4684 &    3 &     3 &    2.8000  &   0.8375  &    0.8 &    26.5453  &   6 &    4  &    1.1375   &  0.9250 \\ \hline
	   	 1.0 &   27.9656 &   2 &     2  &     2.5125 &     0.7125
 & 1.2 &   26.8039 &  6   &    4 &   1.1125 &    0.9250 \\ \hline
	    1.5 &  28.9099 &    1 &   1 &   2.0375 &   0.0125
 &  1.5   &   26.9779 &    7  &  4   & 0.8625   &  0.9250 \\ \hline
	   	\end{tabular}
}
		 \label{HItable-5}
\end{table}
\noindent From Table \ref{HItable-1} it is clear that as $a$  increases for fixed $b$ the minimum Bayes risk increases and as $b$ increases for fixed $a$ the minimum Bayes risk decreases. In Tables \ref{HItable-2}-\ref{HItable-4} we observed that as coefficients $a_0, a_1, a_2, a_3, a_4$ and $a_5$ increases the minimum Bayes risk increases. In Tables \ref{HItable-5}-\ref{HItable-6} if costs $C_s, C_{\tau}$ and $C_r$ increases then the minimum Bayes risk increases and  when the salvage value $r_s$ increases, the  minimum Bayes risk decreases. 
\begin{table}[H]
\caption{The minimum Bayes risk and optimum DSP for Type-I hybrid censoring as $C_r$ or $r_s$ varies }
\vspace{1em}
\resizebox{16.5cm}{0.1\textheight}{
		\begin{tabular}{|c|ccccc|c|ccccc|}\hline
	  $C_r$  & $r( n_0, r_0,\tau_{0},\zeta_0)$ & $n_0$ & $r_0$ & $\tau_{0}$ & $\zeta_0$ & $r_s$ & $r( n_0, r_0, \tau_0, \zeta_0)$ & $n_0$ & $r_0$ & $\tau_{0}$ & $\zeta_0$  \\ \hline
	   25 &   22.7787 &   4 &    3 &   1.5750 &   0.7875 &  0.05 &  26.5544  &  4 &    4  &    3.0250  &    0.9250
 \\ \hline
	    35 &   29.6324 &    6  &    5  &   1.6375  &    1.0375 &   0.10 &  26.5352 &   4 &     4 &   3.0000 &    0.9250 \\ \hline
	  50 &   38.8182  &   8   &    7   &   1.5375  &    1.2500 &    0.20  &   26.4478  &   5  &    4   &    1.6625 &     0.9250\\ \hline
	   	 65 &   47.1201 &  10 &    9 &    1.4500  &   1.4125
 &  0.30  & 26.2983 &   5 &   4 &    1.6375   &  0.9250 \\ \hline
	   85  &   57.2562 &  12  &   11  &    1.3750  &   1.5625
 &  0.35 &  26.2220 &   6  &  4 &    1.1375   & 0.9250 \\ \hline
	   	\end{tabular}
}
		 \label{HItable-6}
\end{table} 
\par
For Type-I censoring, we present the optimum DSP for the fifth degree polynomial, with the following standard set of hyper parameters, coefficients and costs: $a=1.5, b=0.8, a_0 = 2, a_1 = 2, a_2 = 2, a_3=2, a_4=2, a_5=2, C_r = 30, C_s = 0.5, r_s=0, C_\tau = 0.5, \zeta^{*}=6$. In Tables \ref{Ttable-1}-\ref{Ttable-6} the values of the different hyper parameters or coefficients or costs are given in columns $1$ and $6$.  The minimum Bayes risk is denoted by $r( n_0, \tau_{0},\zeta_0)$ and the corresponding sampling plan is $( n_0, \tau_{0},\zeta_0)$.

\begin{table}[H]
\caption{ The minimum Bayes risk and optimum DSP for Type-I censoring as $a$ and $b$ varies}
\vspace{1em}
\centering
\resizebox{8.5cm}{0.1\textheight}{
		\begin{tabular}{|cc|cccc|}\hline
	  $a$ & $b$  & $r( n_0, \tau_{0},\zeta_0)$ & $n_0$ &  $\tau_{0}$ & $\zeta_0$  \\ \hline
	 1.5 &    0.8  &  27.0038 &   5 &    1.7000 & 0.9375  \\ \hline
	  1.5 &   1.2 &   22.9851 &   6 &   1.6875  &   1.0750  \\ \hline
	 2.5 &   2.5 &   21.1783  &  6  &    1.6125 & 1.2250\\ \hline
	 3.0 & 2.5  & 24.8622 & 6 & 1.7000 & 1.1250   \\ \hline
	 3.0 & 3.0 & 21.4133 & 6 & 1.5750 & 1.2625\\ \hline
	   	\end{tabular}
}
		 \label{Ttable-1}
\end{table}
\begin{table}[H]
\caption{ The minimum Bayes risk and optimum DSP for Type-I censoring as $a_0$ or $a_1$ varies }
\vspace{1em}
\resizebox{16.5cm}{0.1\textheight}{
		\begin{tabular}{|c|cccc|c|cccc|}\hline
	  $a_0$  & $r( n_0, \tau_{0},\zeta_0)$ & $n_0$ &  $\tau_{0}$ & $\zeta_0$ & $a_1$ & $r( n_0, \tau_0, \zeta_0)$ & $n_0$ & $\tau_{0}$  & $\zeta_0$  \\ \hline
		\multicolumn{1}{|c|}{0.5} &   26.5210 &   5 &   1.7125 &   0.9500 & \multicolumn{1}{c|}{0.5} &  26.7003 &   5 &  1.7000 &   0.9500\\ \hline
	    \multicolumn{1}{|c|}{1.0} &  26.6833  &  5  &  1.7000 &   0.9375 &  \multicolumn{1}{c|}{1.0} & 26.8029 &    5 &    1.7000  &  0.9500 \\ \hline
		\multicolumn{1}{|c|}{1.5} & 26.8436    & 5  & 1.7000   & 0.9375 & 	\multicolumn{1}{c|}{1.5} & 26.9035 &   5 &    1.7000 &   0.9375 \\ \hline
		\multicolumn{1}{|c|}{2.0} & 27.0038  &  5 &    1.7000 &   0.9375 & 	\multicolumn{1}{c|}{2.0} & 27.0038 &    5 &    1.7000 &   0.9375\\ \hline
		\multicolumn{1}{|c|}{2.5} & 27.1626 &   5 &    1.6875 &   0.9250 & 	\multicolumn{1}{c|}{2.5} & 27.1026 & 5 &   1.7000 &    0.9250
 \\ \hline
	   	\end{tabular}
}
		 \label{Ttable-2}
\end{table}

\begin{table}[H]
\caption{  The minimum Bayes risk and optimum DSP for Type-I censoring as $a_2$ or $a_3$ varies }
\vspace{1em}
\resizebox{16.5cm}{0.1\textheight}{
		\begin{tabular}{|c|cccc|c|cccc|}\hline
	  $a_2$  & $r( n_0, \tau_{0},\zeta_0)$ & $n_0$ &  $\tau_{0}$ & $\zeta_0$ & $a_3$ & $r( n_0, \tau_0, \zeta_0)$ & $n_0$ & $\tau_{0}$  & $\zeta_0$  \\ \hline
		\multicolumn{1}{|c|}{0.5} &   26.7282 &    5 &    1.6875 &    0.9500& \multicolumn{1}{c|}{0.5} &   26.6814 &    5 &  1.6750 &    0.9625\\ \hline
	    \multicolumn{1}{|c|}{1.0} & 26.8216  &    5 &    1.6875 &    0.9500 &  \multicolumn{1}{c|}{1.0} &  26.7930   & 5  &   1.6750 &    0.9500 \\ \hline
		\multicolumn{1}{|c|}{1.5} & 26.9135 &    5 &    1.6875 &    0.9375 & 	\multicolumn{1}{c|}{1.5} &  26.9006 &  5 &    1.6875 &    0.9375 \\ \hline
		\multicolumn{1}{|c|}{2.0} & 27.0038 &    5 &    1.7000 &    0.9375 & 	\multicolumn{1}{c|}{2.0} & 27.0038 &    5 &    1.7000 &   0.9375\\ \hline
		\multicolumn{1}{|c|}{2.5} & 27.0919 &    5 &    1.7000 &    0.9250 & 	\multicolumn{1}{c|}{2.5} & 27.1033  &  5  &    1.7000 &   0.9250
 \\ \hline
	   	\end{tabular}
}
		 \label{Ttable-3}
\end{table}

\begin{table}[H]
\caption{  The minimum Bayes risk and optimum DSP for Type-I censoring as $a_4$ or $a_5$ varies }
\vspace{1em}
\resizebox{16.5cm}{0.1\textheight}{
		\begin{tabular}{|c|cccc|c|cccc|}\hline
	  $a_4$  & $r( n_0, \tau_{0},\zeta_0)$ & $n_0$ &  $\tau_{0}$ & $\zeta_0$ & $a_5$ & $r( n_0, \tau_0, \zeta_0)$ & $n_0$ & $\tau_{0}$  & $\zeta_0$  \\ \hline
		\multicolumn{1}{|c|}{0.5} &   26.5349   &  5  &    1.6375 &    0.9875 & \multicolumn{1}{c|}{0.5} &   26.0954 &   5 &    1.5375 &    1.0750 \\ \hline
	    \multicolumn{1}{|c|}{1.0} & 26.7044 &   5 &    1.6625 &    0.9750 &  \multicolumn{1}{c|}{1.0} &  26.4724 &    5 &    1.6000 & 1.0125 \\ \hline
		\multicolumn{1}{|c|}{1.5} & 26.8598 &    5 &    1.6750 & 0.9500 & 	\multicolumn{1}{c|}{1.5} &  26.7653    & 5 &   1.6500 &    0.9625  \\ \hline
		\multicolumn{1}{|c|}{2.0} & 27.0038 &    5 &    1.7000 &    0.9375 & 	\multicolumn{1}{c|}{2.0} & 27.0038 &    5 &    1.7000 &   0.9375\\ \hline
		\multicolumn{1}{|c|}{2.5} & 27.1370 &   5 &    1.7125 &   0.9125 & 	\multicolumn{1}{c|}{2.5} & 27.2053    & 5 &    1.7375 &    0.9000
 \\ \hline
	   	\end{tabular}
}
		 \label{Ttable-4}
\end{table}
\noindent From Tables \ref{Ttable-2}-\ref{Ttable-4}  we observed that as coefficients of acceptance cost  $a_0, a_1, a_2, a_3, a_4$ and $a_5$ increase, the minimum Bayes risk $r(n_0, \tau_{0},\zeta_0)$ also increases and the optimum value of $\zeta_0$ decreases. It is also observed that the optimum value of $\tau_0$ increases as coefficient $a_2, a_3, a_4$ and $a_5$ increases. In Tables \ref{Ttable-5}-\ref{Ttable-6} costs $C_s, C_{\tau}, C_r$ and $r_s$ are varies for different values and we observed that behaviour of minimum Bayes risk and  optimum DSP are as expected.
\begin{table}[H]
\caption{  The minimum Bayes risk and optimum DSP for Type-I censoring as $C_s$ or $C_{\tau}$ varies }
\vspace{1em}
\resizebox{16.5cm}{0.1\textheight}{
		\begin{tabular}{|c|cccc|c|cccc|}\hline
	  $C_s$  & $r( n_0, \tau_{0},\zeta_0)$ & $n_0$ &  $\tau_{0}$ & $\zeta_0$ & $C_{\tau}$ & $r( n_0, \tau_0, \zeta_0)$ & $n_0$ & $\tau_{0}$  & $\zeta_0$  \\ \hline
		\multicolumn{1}{|c|}{0.2} &   25.0552  &  9 &  1.4750 &    1.0750 & \multicolumn{1}{c|}{0.2} &   26.3960 &   5 &    2.5500 & 0.9750 \\ \hline
	    \multicolumn{1}{|c|}{0.3} & 25.8550  &  7 &    1.6375 &   1.0250 &  \multicolumn{1}{c|}{0.5} &   27.0038 &   5 &    1.7000 &    0.9375  \\ \hline
		\multicolumn{1}{|c|}{0.5} & 27.0038  &  5 & 1.7000 &    0.9375 & 	\multicolumn{1}{c|}{0.8} &  27.4884 &   5 &    1.5500 &    0.9250  \\ \hline
		\multicolumn{1}{|c|}{0.8} & 28.2251 &   3 &    2.5250 &    0.8375 & 	\multicolumn{1}{c|}{1.2} &  28.0462  & 6 &    1.1250 &   0.9250 \\ \hline
		\multicolumn{1}{|c|}{1.2} & 29.0845  &  2 & 2.3000 &    0.7125 & 	\multicolumn{1}{c|}{1.5} & 28.3763 &    6 &   1.0750 & 0.9250
 \\ \hline
	   	\end{tabular}
}
		 \label{Ttable-5}
\end{table} 

\begin{table}[H]
\caption{  The minimum Bayes risk and optimum DSP for Type-I censoring as $C_r$ or $r_s$ varies }
\vspace{1em}
\resizebox{16.5cm}{0.1\textheight}{
		\begin{tabular}{|c|cccc|c|cccc|}\hline
	  $C_r$  & $r( n_0, \tau_{0},\zeta_0)$ & $n_0$ &  $\tau_{0}$ & $\zeta_0$ & $r_s$ & $r( n_0, \tau_0, \zeta_0)$ & $n_0$ & $\tau_{0}$  & $\zeta_0$  \\ \hline
		\multicolumn{1}{|c|}{25} &    23.4949 &    4 &    1.5875 &    0.7875 & \multicolumn{1}{c|}{0.05} &   26.9583 &  5 &    1.6750  &   0.9375  \\ \hline
	    \multicolumn{1}{|c|}{50} & 39.5495 &   7 &    1.8250 &    1.2250 &  \multicolumn{1}{c|}{0.10} &   26.9121  &  5  &   1.6625 &    0.9375\\ \hline
		\multicolumn{1}{|c|}{85} & 58.1138 &  11 &    1.7875 &    1.5625 & 	\multicolumn{1}{c|}{0.20} &    26.8185  &   5 &  1.6250  &  0.9375 \\ \hline
		\multicolumn{1}{|c|}{100} & 65.2465  & 12   & 1.7750  &  1.6500 & 	\multicolumn{1}{c|}{0.30} &  26.7229 &   5 &   1.6000 &   0.9250  \\ \hline
		\multicolumn{1}{|c|}{125} & 76.3677  & 14   &  1.7500  &  1.7875 & 	\multicolumn{1}{c|}{0.40} &  26.6071 &   6 &  1.2625   & 0.9375\\ \hline
	   	\end{tabular}
}
		 \label{Ttable-6}
\end{table}

\subsubsection{\sc Non Polynomial Loss Function}
Since the form of the loss function can vary so in this section, we present the optimum DSP for the non polynomial loss function 
considered in Section \ref{S_6}. For the Type-I hybrid censoring, the following standard set of hyper parameters, coefficients and costs: $a=2.5, b=0.8, a_0 = 2, a_1 = 2, a_2 = 2, C_r = 30, C_s = 0.5, r_s=0.3, C_\tau = 5.0, \zeta^{*}=6$ are used. In Tables \ref{NHtable-1}-\ref{NHtable-5} the values of the different hyper parameter or coefficients or costs are given in column $1$ and $7$ and the others 
are kept fixed.

\begin{table}[H]
\caption{The minimum Bayes risk and the optimum DSP for Type-I hybrid censoring as $a$ and $b$ varies }
\vspace{1em}
\centering
\resizebox{8.5cm}{0.1\textheight}{
		\begin{tabular}{|cc|ccccc|}\hline
	  $a$ & $b$  & $r( n_0, r_0, \tau_{0},\zeta_0)$ & $n_0$ & $r_0$ &  $\tau_{0}$ & $\zeta_0$  \\ \hline
        0.2 &   0.2 &  10.5326 &    5 &    2 &    0.1750 &    2.3125  \\ \hline
	    1.5 &   0.4 &  27.4453 &    6 &    3 &    0.3125 &   1.9375 \\ \hline
	    1.5  &  0.8 & 20.2414 &  8 &  4 & 0.2250 & 2.6125 \\ \hline
	    2.5 &   0.8 & 28.4481  & 6 &  3 & 0.3125  & 1.9625   \\ \hline
	    3.0 &  1.5 & 22.6152  &  6 &    3 &    0.1875 & 2.9500 	\\ \hline
	   	\end{tabular}
}
		 \label{NHtable-1}
\end{table}

\begin{table}[H]
\caption{ The minimum Bayes risk and optimum DSP for Type-I hybrid censoring as $a_0$ or $a_1$ varies}
\vspace{1em}
\resizebox{16.5cm}{0.1\textheight}{
		\begin{tabular}{|c|ccccc|c|ccccc|}\hline
	  $a_0$  & $r( n_0, r_0,\tau_{0},\zeta_0)$ & $n_0$ & $r_0$ & $\tau_{0}$ & $\zeta_0$ & $a_1$ & $r( n_0, r_0, \tau_0, \zeta_0)$ & $n_0$ & $r_0$ & $\tau_{0}$ & $\zeta_0$  \\ \hline
	     0.5 &  27.9446  &   6    &   3   &   0.3000  &   2.0375 &  0.5  &   27.5833  &   6   &    3  &    0.2875 &   2.1375
 \\ \hline
	     1.0 &   28.1152 &   6  &    3  &    0.3000 &  2.0125 &  1.0  &   27.8877  &    6  &    3  &    0.3000  &    2.0750 \\ \hline
	    1.5   &  28.2840  &   6   &   3  &    0.3125  &   1.9875 &    1.5   &   28.1737  &  6   &    3  &     0.3000  &   2.0250 \\ \hline
	   	 2.0  &  28.4481  &   6   &   3  &    0.3125  &   1.9625
 &  2.0  &   28.4481  &   6   &    3   &   0.3125  &  1.9625 \\ \hline
   2.5   & 	  28.6111  &    6  &   3   &    0.3125  &  1.9375 
 &   2.5  &   28.7106  &   6  &    3   &    0.3250  &  1.9125 \\ \hline
	   	\end{tabular}
}
		 \label{NHtable-2}
\end{table}

\begin{table}[H]
\caption{ The minimum Bayes risk and optimum DSP for Type-I hybrid censoring as $a_2$ or $C_{\tau}$ varies }
\vspace{1em}
\resizebox{16.5cm}{0.1\textheight}{
		\begin{tabular}{|c|ccccc|c|ccccc|}\hline
	  $a_2$  & $r( n_0, r_0,\tau_{0},\zeta_0)$ & $n_0$ & $r_0$ & $\tau_{0}$ & $\zeta_0$ & $C_{\tau}$ & $r( n_0, r_0, \tau_0, \zeta_0)$ & $n_0$ & $r_0$ & $\tau_{0}$ & $\zeta_0$  \\ \hline
	    0.5 &  21.4982  &  5  & 3  &   0.1500  &   1.5750 & 0.5 &  27.2156  &  4   &   4   &    1.3000  &   2.1000
 \\ \hline
	     1.0 & 25.4359 &   6  &   3 &   0.2000  &   2.9750 & 1.5 &  27.6168 &    4  &   3  &   0.6000  &   1.9625\\ \hline
	     1.5 &  27.3171  &   6 &    3 &    0.2625  &   2.3125  &    3.0 & 28.0288  & 5  &   3 &  0.4125 & 1.9625 \\ \hline
	   	 2.0 & 28.4481  &  6   &    3  &    0.3125  &   1.9625 & 4.0  &   28.2477  &  6  &    3  &   0.3125 &    1.9625 \\ \hline
	     2.5 &  29.1885  &   5  &    2  &    0.2875 & 1.5250 &  5.0  &  28.4481  &   6  &  3  &    0.3125  &   1.9625\\ \hline
	   	\end{tabular}
}
		 \label{NHtable-3}
\end{table}
\noindent  From Table \ref{NHtable-1} it is clear that as $a$  increases, for fixed $b$, the  minimum Bayes risk increases and as $b$ 
increases, for fixed $a$, the minimum Bayes risk decreases. In Tables \ref{NHtable-2}-\ref{NHtable-3} when coefficient $a_0, a_1$ and $a_2$ increases then minimum Bayes risk  $r( n_0, r_0, \tau_{0},\zeta_0)$ increases. The optimum value of $\tau_0$  increases  and $\zeta_0$ decreases as $a_0$ and $a_1$ increase. In Tables \ref{NHtable-3}-\ref{NHtable-4} when costs $C_s, C_{\tau}$ and $C_r$ increase, then 
the minimum Bayes risk increases.  In Table \ref{NHtable-5} when the salvage value $r_s$ increases then the minimum Bayes risk 
decreases as expected and optimal sample sizes $n_0$ and $r_0$ increase.
\begin{table}[H]
\caption{ The minimum Bayes risk and optimum DSP for Type-I hybrid censoring as $C_s$ or $C_r$ varies }
\vspace{1em}
\resizebox{16.5cm}{0.1\textheight}{
		\begin{tabular}{|c|ccccc|c|ccccc|}\hline
	  $C_s$  & $r( n_0, r_0,\tau_{0},\zeta_0)$ & $n_0$ & $r_0$ & $\tau_{0}$ & $\zeta_0$ & $C_r$ & $r( n_0, r_0, \tau_0, \zeta_0)$ & $n_0$ & $r_0$ & $\tau_{0}$ & $\zeta_0$  \\ \hline
	    0.4 &   27.7042  &  10 &   4  &    0.2250 & 2.1000 & 25 &   24.8091 &    5   & 2   &    0.2875 &    1.5125
 \\ \hline
	     0.5 &   28.4481 &   6  &    3  &    0.3125 & 1.9625 & 35 &    31.6670  &   8   &    4  &    0.2750 &    2.3250\\ \hline
	   0.6  & 28.9501  &  4   &    2   &   0.3250   &  1.7375  &   50  & 39.5133  &  10   &    6   &    0.2750  &   3.1000 \\ \hline
	   	 0.7  &  29.3501 &   4   &    2  &   0.3250  &  1.7375 
 & 65 &     45.5177 &  11  &    7  &    0.2625  &  3.6875 \\ \hline
	    0.8 &  29.6455  &  2 &   1   &    0.3625  &   0.0125 
 &  85 &  51.6634 &  12 &    8  & 0.2375 &  4.3625 \\ \hline
	   	\end{tabular}
}
		 \label{NHtable-4}
\end{table}

\begin{table}[H]
\caption{ The minimum Bayes risk and optimum DSP for Type-I hybrid censoring as $r_s$ varies}
\vspace{1em}
\centering
\resizebox{8.5cm}{0.1\textheight}{
		\begin{tabular}{|c|ccccc|}\hline
	  $r_s$  & $r( n_0, r_0, \tau_{0},\zeta_0)$ & $n_0$ & $r_0$ &  $\tau_{0}$ & $\zeta_0$  \\ \hline
		0.05 & 29.1184  &  3   &    2   & 0.4750  &   1.7375 \\ \hline
	    0.10 & 29.0215  &  4   &    2   &    0.3375 &     1.7375 \\ \hline
		0.20 &  28.7866 &   4  &    2  &    0.3375  &    1.7375 \\ \hline
	    0.30 &  28.4481 &  6   &    3   &   0.3125  &  1.9625  \\ \hline
		0.40 & 27.9798  &  8  &    3  &    0.2125  &  1.9625 	\\ \hline
	   	\end{tabular}
}
		 \label{NHtable-5}
\end{table}
 For the Type-I censoring, we also present the optimum DSP for the non polynomial loss function considered in Section \ref{S_6}, with the following standard set of hyper parameters, coefficients and costs: $a=2.5, b=0.8, a_0 = 2, a_1 = 2, a_2 = 2, C_r = 30, C_s = 0.5, r_s=0, C_\tau = 0.5, \zeta^{*}=6$. Numerical results are given in Tables \ref{NTtable-1}-\ref{NTtable-5} where only hyper parameters $a$ and $b$ or one coefficient or one cost is varying and others are kept fixed as defined above.

\begin{table}[H]
\caption{The minimum Bayes risk and optimum DSP for Type-I censoring as $a$ and $b$ varies }
\vspace{1em}
\centering
\resizebox{8.5cm}{0.1\textheight}{
		\begin{tabular}{|cc|cccc|}\hline
	  $a$ & $b$  & $r( n_0, \tau_{0},\zeta_0)$ & $n_0$ &  $\tau_{0}$ & $\zeta_0$  \\ \hline
	 1.5 &    0.4 &   26.6262 &    3 &   1.1125 & 1.9375  \\ \hline
	  1.5 &    0.8 &   19.4142 &   4 &   0.9000 &   2.6125  \\ \hline
	 2.5 &    0.8 &   27.5603  &  4 &    1.0750 &   2.0625 \\ \hline
	 2.5 &    1.2 &   22.2069  &  4 &    0.8875  &  2.6500 \\ \hline
	 3.0 &    1.5 &   21.8535  &  4 &    0.8250 &  2.8750\\ \hline
	   	\end{tabular}
}
		 \label{NTtable-1}
\end{table}
\begin{table}[H]
\caption{ The minimum Bayes risk and optimum DSP for Type-I censoring as $a_0$ or $a_1$ varies }
\vspace{1em}
\resizebox{16.5cm}{0.1\textheight}{
		\begin{tabular}{|c|cccc|c|cccc|}\hline
	  $a_0$  & $r( n_0, \tau_{0},\zeta_0)$ & $n_0$ &  $\tau_{0}$ & $\zeta_0$ & $a_1$ & $r( n_0, \tau_0, \zeta_0)$ & $n_0$ & $\tau_{0}$  & $\zeta_0$  \\ \hline
		\multicolumn{1}{|c|}{0.5} &  27.0238 &   4 &    1.0750 &    2.1375 & \multicolumn{1}{c|}{0.5} &  26.6463 &    4    &    1.0375  & 2.2500 \\ \hline
	    \multicolumn{1}{|c|}{1.0} &  27.2050 &    4   &    1.0750  &   2.1125 &  \multicolumn{1}{c|}{1.0} &  26.9657  &   4  &    1.0500  &   2.1750  \\ \hline
		\multicolumn{1}{|c|}{1.5} & 27.3838 &   4   &   1.0750   & 2.0875 & 	\multicolumn{1}{c|}{1.5} & 27.2702   &  4   &    1.0625 &   2.1125  \\ \hline
		\multicolumn{1}{|c|}{2.0} & 27.5603   &  4   &   1.0750   & 2.0625 & 	\multicolumn{1}{c|}{2.0} & 27.5603 &   4   &   1.0750  &    2.0625 \\ \hline
		\multicolumn{1}{|c|}{2.5} & 27.7262  &   3   &   1.1125  &  1.9375 & 	\multicolumn{1}{c|}{2.5} & 27.8216 &   3  &   1.1250  &   1.9125
 \\ \hline
	   	\end{tabular}
}
		 \label{NTtable-2}
\end{table}

\begin{table}[H]
\caption{ The minimum Bayes risk and optimum DSP for Type-I censoring as $a_2$ or $C_s$ varies}
\vspace{1em}
\resizebox{16.5cm}{0.1\textheight}{
		\begin{tabular}{|c|cccc|c|cccc|}\hline
	  $a_2$  & $r( n_0, \tau_{0},\zeta_0)$ & $n_0$ &  $\tau_{0}$ & $\zeta_0$ & $C_s$ & $r( n_0, \tau_0, \zeta_0)$ & $n_0$ & $\tau_{0}$  & $\zeta_0$  \\ \hline
		\multicolumn{1}{|c|}{0.5} &    20.9985  &   4 &    0.5375   & 4.7375 & \multicolumn{1}{c|}{0.2} &   25.9956  &    8  &    0.8750  &   2.3000 \\ \hline
	    \multicolumn{1}{|c|}{1.0} & 24.5967  &  4 &    0.8000  &   3.0500 &  \multicolumn{1}{c|}{0.3} &  26.6479 &   6  &   0.9250  &   2.2000\\ \hline
		\multicolumn{1}{|c|}{1.5} & 26.4246  &  4  &  0.9625 &   2.4125 & 	\multicolumn{1}{c|}{0.5} &  27.5603 &  4  &    1.0750  &    2.0625 \\ \hline
		\multicolumn{1}{|c|}{2.0} & 27.5603 &    4   &    1.0750  &  2.0625 & 	\multicolumn{1}{c|}{0.8} & 28.3770 &   2    &    0.9500  &    1.7375 \\ \hline
		\multicolumn{1}{|c|}{2.5} & 28.3162  &   3   &    1.2250  &    1.7375  & 	\multicolumn{1}{c|}{1.2} & 29.1411  &   1  &   0.7250 &    0.6000
 \\ \hline
	   	\end{tabular}
}
		 \label{NTtable-3}
\end{table}
\noindent It is clear from the Tables \ref{NTtable-2}-\ref{NTtable-3} that the minimum Bayes risk increases as the coefficient $a_0, a_1$ and $a_2$ increase. In Table \ref{NTtable-3} when the cost $C_s$ increases the minimum Bayes risk increases and $n_0$ decreases. In Table \ref{NTtable-4} as the cost $C_{\tau}$ increases then the minimum Bayes risk increases and $\tau_0$ decreases. When the cost $C_r$ increases then the minimum Bayes risk increases, $n_0$ and $\tau_0$ increase and $\tau_0$ decreases.
\begin{table}[H]
\caption{ The minimum Bayes risk and optimum DSP for Type-I censoring as $C_{\tau}$ or $C_r$ varies}
\vspace{1em}
\resizebox{16.5cm}{0.1\textheight}{
		\begin{tabular}{|c|cccc|c|cccc|}\hline
	  $C_{\tau}$  & $r( n_0, \tau_{0},\zeta_0)$ & $n_0$ &  $\tau_{0}$ & $\zeta_0$ & $C_r$ & $r( n_0, \tau_0, \zeta_0)$ & $n_0$ & $\tau_{0}$  & $\zeta_0$  \\ \hline
		\multicolumn{1}{|c|}{0.2} &   27.2069  &  4  &    1.3000 &   2.0875 & \multicolumn{1}{c|}{25} &   24.0664 & 2   &  1.0625  &   1.5125 \\ \hline
	    \multicolumn{1}{|c|}{0.5} & 27.5603  &  4   &    1.0750 &    2.0625 &  \multicolumn{1}{c|}{35} &  30.6915 &     4  &    1.0500  &    2.3125 \\ \hline
		\multicolumn{1}{|c|}{0.7} & 27.7625 &   4    &    0.9500 &   2.0250 & 	\multicolumn{1}{c|}{50} &  38.4988 &    6  &    0.9250  &   3.0875  \\ \hline
		\multicolumn{1}{|c|}{1.0} & 28.0240 &   4   &    0.7875  &    1.9875  & 	\multicolumn{1}{c|}{65} & 44.6010 &    7 &    0.8500  &  3.6875 \\ \hline
		\multicolumn{1}{|c|}{1.5} & 28.3421  &  4    &    0.6000  &    1.9625 & 	\multicolumn{1}{c|}{85} &  50.9093  &  8 &    0.7750   &   4.3625
 \\ \hline
	   	\end{tabular}
}
		 \label{NTtable-4}
\end{table}
\begin{table}[H]
\caption{ The minimum Bayes risk and optimum DSP for Type-I censoring as $r_s$  varies }
\vspace{1em}
\centering
\resizebox{8.5cm}{0.1\textheight}{
		\begin{tabular}{|c|cccc|}\hline
	  $r_s$  & $r( n_0, \tau_{0},\zeta_0)$ & $n_0$ &  $\tau_{0}$ & $\zeta_0$  \\ \hline
		0.05 &   27.5361  &  4   &    1.0500 & 2.0500  \\ \hline
	    0.10 & 27.5112  &  4  &   1.0375  &    2.0500  \\ \hline
		0.20 & 27.4589  &  4   &    0.9875  &    2.0375  \\ \hline
	    0.30 &  27.4025  &  4 &    0.9125  &  2.0125  \\ \hline
		0.40 & 27.3100  &  5 &   0.6875  &  2.1000 	\\ \hline
	   	\end{tabular}
}
		 \label{NTtable-5}
\end{table}
\noindent From Table \ref{NTtable-5} it is that as the salvage value $r_s$ increases, the  minimum Bayes risk and the $\tau_0$ decrease.
\section{\sc Conclusion}
\label{S_7}
In this work, we have considered the sampling plan in the life testing experiment under Type-I and  Type-I hybrid  censoring scheme where lifetimes are exponentially distributed with parameter $\lambda$. We have proposed that a decision theoretic sampling plan (DSP) can be obtained by using a suitable estimator of $\lambda$, in place of the estimator of mean lifetime $\theta=\frac{1}{\lambda}$. The proposed estimator of $\lambda$ always exists for censored samples. Moreover, we have developed a methodology for finding a DSP using a decision function based on this estimator of  $\lambda$ under Type-I and Type-I hybrid censoring. Numerically it is observed that the optimum DSP is better than sampling plans of \cite{Lam1994bayesian}, \cite{lin2008-10exact} and as good as a Bayesian sampling plan in terms of Bayes risk for Type-I and  Type-I hybrid censoring.  The main advantage of our study is that the proposed sampling plan can be used
quite conveniently for higher degree polynomial and for non-polynomial loss functions without any additional effort as compared to the 
existing BSP. 

\section*{\sc Acknowledgements:} The authors would like to thank two unknown reviewers and the Associate Editor for their constructive 
comments which have helped to improve the manuscript significantly.

\section{ \sc Appendix}
\label{S_8}
\subsection{\sc Proof of Theorem \ref{T_3.1}}
\begin{proof}
The Bayes risk of DSP with respect to the loss function (\ref{new-loss}) is given by  
\bea
 r(n, \tau, \zeta)
            & = & n(C_s-r_s)+ E(M)r_{s} + \tau C_{\tau} + a_0 + a_1 \mu_1 + a_2 \mu_2 \nonumber \\
            &  & \quad + \int_{0}^{\infty}(C_r - a_0 - a_1 \lambda-a_2 \lambda ^2)P(\widehat{\lambda} \geq \zeta)\frac{b^a}{\Gamma (a) }\lambda^{a-1}e^{-\lambda b} d\lambda \nonumber \\
            & = & n(C_{s}-r_{s}) + E_{\lambda}E_{X/\lambda}(M)r_{s} + \tau C_{\tau}+a_0 +a_1 \mu_1 +a_2 \mu_2 \nonumber\\
            &  & \quad + \sum_{l=0}^{2}C_{l}\frac{b^a}{\Gamma (a)}\int_{0}^{\infty} \lambda^{a+l-1} e^{-\lambda b}P(\widehat{\lambda} \geq \zeta)\ d\lambda,  \label{bayes-risk}
\eea
where $C_l$ is defined as 
\begin{equation}
C_l =
\begin{cases}
C_r- a_l  &\mbox{if} \quad l=0,\\
-a_l      &\mbox{if} \quad l= 1,2.
\end{cases}
\end{equation}
Using Lemma \ref{L_3.1} in (\ref{bayes-risk}) we get 
\bea
&  & \int_{0}^{\infty}\lambda^{a+l-1}e^{-\lambda b}P(\widehat{\lambda} \geq \zeta)\ d\lambda \nonumber \\
&  & = \int_{0}^{\infty}\lambda^{a+l-1}e^{-\lambda (b+n\tau)}\ d\lambda \ \textit{I}_{(\zeta=0)} +\sum_{m=1}^{n}\sum_{j=0}^{m}\binom{n}{m}\binom{m}{j}(-1)^j \nonumber \\
&  & \hspace{3.5cm}\times \int_{0}^{\infty} \int_{\zeta}^{\frac{1}{\tau_{j,m}}}\lambda^{a+l+m-1}\frac{e^{-\lambda\{b+\frac{m}{y} \}}}{y^2}\big(\frac{1}{y}-\tau_{j,m}\big)^{m-1}dy \ d\lambda \nonumber \\
&  & = \frac{\Gamma{(a+l)}}{(b+n\tau)^{(a+l)}} \textit{I}_{(\zeta=0)} + \sum_{m=1}^{n}\sum_{j=0}^{m}\binom{n}{m}\binom{m}{j}(-1)^j\frac{(m)^m}{\Gamma (m)}\int_{\zeta}^{\frac{1}{\tau_{j,m}}}\frac{\big(\frac{1}{y}-\tau_{j,m}\big)^{m-1}\Gamma{(a+l+m)}}{y^2\{b+\frac{m}{y}\}^{a+l+m}}dy \nonumber\\
&  & = \frac{\Gamma{(a+l)}}{(b+n\tau)^{(a+l)}} \textit{I}_{(\zeta=0)} + \sum_{m=1}^{n}\sum_{j=0}^{m}\binom{n}{m}\binom{m}{j}\frac{(m)^m (-1)^j}{\Gamma (m)}\int_{0}^{\frac{1}{\zeta}-\tau_{j,m}}\frac{v^{m-1}\Gamma{(a+l+m)}}{\{b+m\tau_{j,m}+mv\}^{a+l+m}}dv.    
\ \ \ \ \ \ \label{step-2}    
\eea
Using $\ds C_{j,m}=b+m \tau_{j,m}$ in (\ref{step-2}), we can write
\begin{align}
\sum_{m=1}^{n}&\sum_{j=0}^{m}\binom{n}{m}\binom{m}{j}(-1)^j\frac{(m)^m}{\Gamma (m)}\frac{\Gamma{(a+l+m)}}{C_{j,m}^{a+l+m}}\int_{0}^{\frac{1}{\zeta}- \tau_{j,m}}\frac{v^{m-1}}{\Big(1+\frac{mv}{C_{j,m}}\Big)^{a+l+m}}dv \nonumber\\
&=\sum_{m=1}^{n}\sum_{j=0}^{m}\binom{n}{m}\binom{m}{j}(-1)^j\frac{\Gamma{(a+l)}}{(C_{j,m})^{a+l}}\frac{\Gamma{(a+l+m)}}{\Gamma (m)\Gamma{(a+l)}}\int_{0}^{\frac{m(\frac{1}{\zeta}- \tau_{j,m})}{C_{j,m}}}\frac{z^{m-1}}{(1+z)^{a+l+m}}dz.   \label{step-3}
\end{align}
Now taking a transformation $\ds z=u/(1-u)$, we have 
\begin{align*}
    \int_{0}^{C^{*}_{j,m}}\frac{z^{m-1}}{(1+z)^{a+l+m}}dz=\int_{0}^{S^*_{j,m}}u^{m-1}(1-u)^{a+l-1}du=B_{S^*_{j,m}}(m,a+l),
\end{align*} 
where $\ds C^{*}_{j,m}=\frac{m(\frac{1}{\zeta}- \tau_{j,m})}{C_{j,m}}$, \ \ $\ds S^{*}_{j,m}=\frac{C^{*}_{j,m}}{1+C^{*}_{j,m}}$ , \hspace{0.5 cm} and 
\begin{equation}
 B_{x}(\alpha,\beta)=\int_{0}^{x}u^{\alpha -1}(1-u)^{\beta-1}du,  \ \ \ \ \ 0\leq x \leq 1, \nonumber
\end{equation}
is the incomplete beta function. If the cumulative distribution function of the beta distribution is given by $I_x(\alpha,\beta)=B_x(\alpha,\beta)/B(\alpha,\beta),$ 
then using (\ref{step-3}) the Bayes risk is finally obtained as
\begin{align}
r(n,\tau,\zeta) &= n(C_{s}-r_{s}) + E(M)r_{s} +\tau C_{\tau}+a_0 +a_1 \mu_1 +a_2 \mu_2 
 \nonumber \\            
& + \sum_{l=0}^2 C_{l}\frac{b^{a}}{\Gamma (a)}\bigg[\frac{\Gamma{(a+l)}}{(b+n\tau)^{(a+l)}} \textit{I}_{(\zeta=0)}+\sum_{m=1}^{n}\sum_{j=0}^m (-1)^{j}\binom{n}{m}\binom{m}{j} \frac{\Gamma{(a+l)}}{(C_{j,m})^{a+l}}I_{S^*_{j,m}}(m,a+l)\bigg],
\end{align}
\noindent where 
$
E(M) = \sum_{m=1}^{n}\sum_{j=0}^{m}m\binom{n}{m}\binom{m}{j}(-1)^{j}\frac{b^a}{(b+(n-m+j)\tau)^a}  \nonumber 
$.
\end{proof}
In general, for higher degree polynomial i.e for $k>2$, the Bayes risk can be evaluated in a similar way for Type-I censoring.

\subsection{\sc Proof of Theorem \ref{T_3.3}}
\begin{proof}
Note that the Bayes risk can be written as
\begin{align*}
\begin{split}
    r(n, \tau, \zeta) = n(C_s-r_s) + \tau C_{\tau} + E(M)r_{s}+ E_{\lambda}\big\{(a_0 +a_1\lambda+ \ldots +a_k \lambda^k)P(\widehat{\lambda}< \zeta) + C_r P(\widehat{\lambda} \geq \zeta)\big\}.
\end{split}
\end{align*}
Now we know that $a_0+a_1\lambda + \ldots +a_k\lambda^k \geq 0$ and $C_r$, the rejection cost, is non negative. Since 
$(n_0, \tau_0, \zeta_0)$ is the optimal sampling plan so the corresponding  Bayes risk is 
\begin{align}
    r(n_0, \tau_0, \zeta_0) \geq n_0 (C_s-r_s) + \tau_0 C_{\tau}.   \label{eq-14}
\end{align}
Now when $\zeta=0$ we reject the batch without sampling and the corresponding Bayes risk is given by $r(0,0,0) =   C_r.$   
When $ \zeta=\infty$ we accept the batch without sampling and corresponding Bayes risk is given by $r(0,0,\infty) = a_0+a_1\mu_1+ \ldots +a_k\mu_k.$ 
Then the optimal Bayes risk is
\begin{align}
    r(n_0, \tau_0, \zeta_0) \leq min \big\{ r(0,0,0),r(0,0,\infty),r(n, \tau, \zeta')\big\}. \label{eq-15}
\end{align}
Hence from equations (\ref{eq-14}) and (\ref{eq-15}) we have
\begin{align*}
 n_0(C_s-r_s) + \tau_0 C_{\tau} \leq min \big\{ r(0,0,0),r(0,0,\infty),r(n, \tau, \zeta')\big\}.
\end{align*}
from where it follows that 
\begin{align*}
&n_0 \leq min \bigg\{ \frac{C_r}{C_s-r_s},\frac{a_0 + a_1\mu_1 + \ldots + a_k \mu_k}{C_s-r_s},\frac{r(n, \tau, \zeta')}{C_s-r_s}\bigg\} \\
&\tau_0 \leq min \bigg\{ \frac{C_r}{C_{\tau}},\frac{a_0 + a_1\mu_1 + \ldots + a_k \mu_k}{C_{\tau}},\frac{r(n, \tau, \zeta')}{C_{\tau}}\bigg\}.
\end{align*}
\end{proof}
\subsection{\sc Proof of Theorem \ref{T_4.1}}
\begin{proof}
The Bayes risk of DSP with respect to the loss function (\ref{new-loss1}) is given by 
\bea
 r(n, r,\tau, \zeta)
            & = & n(C_s-r_{s}) + E(M)r_{s} + E(\tau^{*}) C_{\tau} + a_0 + a_1 \mu_1 + a_2 \mu_2 \nonumber \\
             &  & \quad + \int_{0}^{\infty}(C_r - a_0 - a_1 \lambda-a_2 \lambda ^2)P(\widehat{\lambda} \geq \zeta)\frac{b^a}{\Gamma (a) }\lambda^{a-1}e^{-\lambda b} d\lambda \nonumber \\
            & = & n(C_{s}-r_{s}) + E(M)r_{s} + E(\tau^{*}) C_{\tau}+a_0 +a_1 \mu_1 +a_2 \mu_2 \nonumber\\
            &  & \quad + \sum_{l=0}^{2}C_{l}\frac{b^a}{\Gamma (a)}\int_{0}^{\infty} \lambda^{a+l-1} e^{-\lambda b}P(\widehat{\lambda} \geq \zeta)\ d\lambda    \label{bayes-risk1}
\eea
where $C_l$ is defined as earlier. Let $\zeta^{*}=max\{\frac{1}{n\tau},\zeta\}$, where $\zeta>0$ and
\bea
 R_{l,j,m} & = & \int_{0}^{\infty}\int_{\zeta^{*}}^{\infty}\lambda^{a+l-1}\frac{e^{-\lambda \{b+\tau(n-m+j)\}}}{y^2}\pi \Big(\frac{1}{y}-\tau_{j,m}; m, m\lambda \Big)dy \ d\lambda     \nonumber\\
  & = & \frac{(m)^m}{\Gamma(m)}\int_{0}^{\infty}\int_{\zeta^{*}}^{\frac{1}{\tau_{j,m}}}\lambda^{a+l+m-1}\frac{e^{-\lambda\{b+\frac{m}{y} \}}}{y^2}\big(\frac{1}{y}-\tau_{j,m}\big)^{m-1}dy \ d\lambda \nonumber \\ 
  & = & \frac{(m)^m}{\Gamma(m)}\int_{\zeta^{*}}^{\frac{1}{\tau_{j,m}}}\frac{\big(\frac{1}{y}-\tau_{j,m}\big)^{m-1}\Gamma{(a+l+m)}}{y^2\{b+\frac{m}{y}\}^{a+l+m}}dy \nonumber\\
  & = & \frac{(m)^m}{\Gamma(m)}\int_{0}^{\frac{1}{\zeta^{*}}-\tau_{j,m}}\frac{v^{m-1}\Gamma{(a+l+m)}}{\{b+m\tau_{j,m}+mv\}^{a+l+m}}dv \nonumber\\
  & = &\frac{(m)^m}{\Gamma(m)}\frac{\Gamma{(a+l+m)}}{C_{j,m}^{a+l+m}}\int_{0}^{\frac{1}{\zeta^{*}}- \tau_{j,m}}\frac{v^{m-1}}{\Big(1+\frac{mv}{C_{j,m}}\Big)^{a+l+m}}dv \nonumber\\
  & = & \frac{\Gamma{(a+l)}}{(C_{j,m})^{a+l}}\frac{\Gamma{(a+l+m)}}{\Gamma (m)\Gamma{(a+l)}}\int_{0}^{\frac{m(\frac{1}{\zeta^{*}}- \tau_{j,m})}{C_{j,m}}}\frac{z^{m-1}}{(1+z)^{a+l+m}}dz, \nonumber
\eea
where $\ds C_{j,m}=b+m \tau_{j,m}$. Now taking a transformation $\ds z=u/(1-u)$, we have 
\begin{align*}
    \int_{0}^{C^{*}_{j,m}}\frac{z^{m-1}}{(1+z)^{a+l+m}}dz=\int_{0}^{S^*_{j,m}}u^{m-1}(1-u)^{a+l-1}du=B_{S^*_{j,m}}(m,a+l),
\end{align*} 
where $\ds C^{*}_{j,m}=\frac{m(\frac{1}{\zeta^{*}}- \tau_{j,m})}{C_{j,m}}$ and \ \ $\ds S^{*}_{j,m}=\frac{C^{*}_{j,m}}{1+C^{*}_{j,m}}$. Using $ B_{x}(\alpha,\beta)$ and $ I_x(\alpha,\beta)$ defined earlier, we obtain the expression
\bea
R_{l,j,m}= \frac{\Gamma{(a+l)}}{(C_{j,m})^{a+l}} 
I_{S^*_{j,m}}(m,a+l). \label{rjm}
\eea
 Using Lemma \ref{L_4.1} in (\ref{bayes-risk1}) and by (\ref{rjm}) we get 
\bea
&  & \int_{0}^{\infty}\lambda^{a+l-1}e^{-\lambda b}P(\widehat{\lambda} \geq \zeta)\ d\lambda \nonumber \\
&  & = \int_{0}^{\infty}\lambda^{a+l-1}e^{-\lambda (b+n\tau)}\ d\lambda \ \textit{I}_{(\zeta=0)} +\sum_{m=1}^{n}\sum_{j=0}^{m}\binom{n}{m}\binom{m}{j}(-1)^j \nonumber \\
&  & \hspace{1.5cm} \times \int_{0}^{\infty}\int_{\zeta^{*}}^{\infty}\lambda^{a+l-1}\frac{e^{-\lambda \{b+\tau(n-m+j)\}}}{y^2} 
\pi \Big(\frac{1}{y}-\tau_{j,m}; m, m\lambda \Big)dy \ d\lambda     \nonumber \\
&  & + \int_{0}^{\infty}\int_{\zeta^{*}}^{\infty}\lambda^{a+l-1}\frac{e^{-\lambda b}}{y^2}
\pi \Big(\frac{1}{y}; r, r\lambda \Big)dy \ d\lambda  + \sum_{k=1}^{r}\binom{n}{r}\binom{r-1}{k-1}(-1)^k \frac{r}{(n-r+k)} \nonumber\\ 
&  & \hspace{1.5cm}\times\int_{0}^{\infty}\int_{\zeta^{*}}^{\infty}\lambda^{a+l-1}\frac{e^{-\lambda \{b+\tau(n-r+k)\}}}{y^2} 
\pi \Big(\frac{1}{y}-\tau_{k,r}; r, r\lambda \Big)dy \ d\lambda  \nonumber \\
&  & = \frac{\Gamma{(a+l)}}{(b+n\tau)^{(a+l)}} \textit{I}_{(\zeta=0)}+\sum_{m=1}^{n}\sum_{j=0}^{m}\binom{n}{m}\binom{m}{j}(-1)^j R_{l,j,m} +R_{l,r-n,r} \nonumber \\
&  & \hspace{6cm} +\sum_{k=1}^{r}\binom{n}{r}\binom{r-1}{k-1}(-1)^k \frac{r}{(n-r+k)} R_{l,k,r}.\nonumber
\eea

\noindent Thus Bayes risk of DSP under Type-I hybrid censoring is given by 
\bea
 r(n, r,\tau, \zeta) & = & n(C_{s}-r_{s}) + E(M)r_{s} + E(\tau^{*}) C_{\tau}+a_0 +a_1 \mu_1 +a_2 \mu_2 \nonumber\\
            &  & \quad + \sum_{l=0}^{2}C_{l}\frac{b^a}{\Gamma (a)}\bigg\{\frac{\Gamma{(a+l)}}{(b+n\tau)^{(a+l)}} \textit{I}_{(\zeta=0)}+\sum_{m=1}^{n}\sum_{j=0}^{m}\binom{n}{m}\binom{m}{j}(-1)^j R_{l,j,m} +R_{l,r-n,r}\nonumber \\
            &  & \hspace{2.5cm}+\sum_{k=1}^{r}\binom{n}{r}\binom{r-1}{k-1}(-1)^k \frac{r}{(n-r+k)} R_{l,k,r}\bigg\},
\eea
\noindent where
\bea
E(M)& = &\sum_{m=1}^{r-1}\sum_{j=0}^{m}m\binom{n}{m}\binom{m}{j}(-1)^{j}\frac{b^a}{(b+(n-m+j)\tau)^a}\nonumber \\ 
& & \hspace{2.5cm}+\sum_{k=r}^{n}\sum_{j=0}^{k}r\binom{n}{k}\binom{k}{j}(-1)^j\frac{b^a}{(b+(n-k+j)\tau)^a} \nonumber \\
E(\tau^{*})& = & r\binom{n}{r}\sum_{j=0}^{r-1}\binom{r-1}{j}(-1)^{r-1-j}\bigg\{\frac{b}{(n-j)^2(a-1)}- \frac{tb^a}{(n-j)((n-j)\tau+b)^a} \nonumber \\
&  & -\frac{b^a}{(n-j)^2(a-1)((n-j)\tau+b)^{a-1}}\bigg\} +  \sum_{k=r}^{n}\sum_{j=0}^{k}\tau\binom{n}{k}\binom{k}{j}(-1)^j\frac{b^a}{(b+(n-k+j)\tau)^a}. \nonumber 
\eea
For computation of $E(M)$ and $E(\tau^{*})$ see \cite{liang2013optimal}.
\end{proof}
In general, for higher degree polynomial, i.e., for $k>2$, the Bayes risk can be evaluated in a similar way for   Type-I hybrid censoring.

\section*{\sc References}
\addcontentsline{toc}{section}{Bibliography}
\small
\begingroup
\renewcommand{\section}[2]{}
\bibliographystyle{plainnat}


\end{document}